\documentclass[letterpaper, 10 pt, conference]{ieeeconf}  

\IEEEoverridecommandlockouts                              

\usepackage{graphicx}          
\usepackage{fontenc}
\usepackage{graphicx}
\usepackage{epstopdf}
\usepackage{subfigure}
\usepackage{array, multirow}
\usepackage{psfrag, tabularx}
\usepackage{amssymb, amsbsy, bm, amsfonts}
\usepackage{longtable,ctable}
\usepackage{soul,psfrag}
\usepackage{algorithm}
\usepackage{etoolbox}

\usepackage{algorithmic}

\AtBeginEnvironment{algorithmic}{}

\usepackage{mathrsfs}
\usepackage{indentfirst}
\usepackage{wrapfig}
\usepackage{tabularx}
\usepackage{tabularx}
\usepackage[tbtags]{amsmath}
\usepackage{doi}

\newtheorem{assumption}{Assumption}[section]
\newtheorem{definition}{Definition}[section]

\newtheorem{proposition}{Proposition}[section]
\newtheorem{lemma}{Lemma}[section]
\newtheorem{remark}{Remark}[section]

\newcommand\propref[1]{Proposition~\ref{#1}}

\newcommand\figref[1]{Figure~\ref{#1}}
\newcommand\figsref[1]{Figures~\ref{#1}}

\newcommand\defref[1]{Definition~\ref{#1}}
\newcommand\assref[1]{Assumption~\ref{#1}}

\newcommand\algref[1]{Algorithm~\ref{#1}}

\newcommand\lemref[1]{Lemma~\ref{#1}}

\title{\LARGE \bf{ 
		Joint Observer Gain and Input Design for Asymptotic Active Fault Diagnosis}}

\author{Feng Xu$^{*}$, Yiming Wan, Ye Wang and Vicen\c{c} Puig
	 \thanks{Corresponding author: Feng Xu (xu.feng@sz.tsinghua.edu.cn).}
	\thanks{Feng Xu is with Tsinghua Shenzhen International Graduate School, Tsinghua University, Shenzhen 518055, P.R.China.}
	\thanks{Yiming Wan is with School of Artificial Intelligence and Automation, Huazhong University of Science and Technology, and Key Laboratory of Image Processing and Intelligent Control, Ministry of Education China, Wuhan 430074, P.R.China.}
	\thanks{Ye Wang is with Department of Electrical and Electronic Engineering, The University of Melbourne, Parkville VIC 3010, Australia.}
	\thanks{Vicen\c c Puig is with Institut de Rob\`{o}tica i Inform\`{a}tica Industrial (CSIC-UPC), Universitat Polit\`ecnica de Catalunya, Barcelona 08028, Spain.}
}

\begin{document}

	\maketitle
	\thispagestyle{empty}
	\pagestyle{empty}


\begin{abstract}                          
	This paper proposes a joint gain and input design method for observer-based asymptotic active fault diagnosis, which is based on a newly-defined notion named the excluding degree of the origin from a zonotope. Using the excluding degree, a quantitative specification is obtained to characterize the performance of set-based robust fault diagnosis. Furthermore, a single gain design method and a joint gain and input design method are proposed, respectively. This is the first work to achieve a joint observer gain and input design for set-based active fault diagnosis. Compared with the existing methods that design gains and input separately, the proposed joint gain and input design method has advantages to exploit the fault diagnosis potential of observer-based schemes. Finally, several examples are used to illustrate the effectiveness of the proposed methods.
\end{abstract}

	\IEEEpeerreviewmaketitle
	
	\section{Introduction}
	\label{Section1}
	
	As the complexity of technical systems increases, fault diagnosis (FD) plays an important role in improving their safety and reliability. Since the operation of technical systems is always under the effect of system uncertainties such as process disturbances, measurement noises, etc., the performance  of FD methods is deeply affected by uncertainties. For any FD method, if it cannot deal with uncertainties in a systematic way, its effectiveness in a real FD task is not trustable. In the literature, the FD ability of dealing with uncertainties is called robustness. Two classes of robust FD methods are proposed in the literature based on the ways to model uncertainties: the stochastic and set-based methods. The stochastic methods consider that uncertainties are modeled by means of known stochastic distributions and then use the probability theory as the mathematical tool to implement FD \cite{DOHLER2016244,WanTAC20188013740,ShangAuto2021109434,WanRNC2021}. Differently, the set-based methods are deterministic approaches by considering that uncertainties are bounded by known sets, e.g., zonotopes, intervals, polytopes, ellipsoids, etc. \cite{Blanchini2015}, and leverages set theory to implement output consistency test with explicit uncertainty propagation bounds for robust FD. However, both classes of methods have their own advantages and disadvantages, which should be chosen based on particular applications by considering specific features of uncertainties. 
	
	This paper focuses on set-based FD, which also includes two classes, i.e., the passive FD (PFD) and active FD (AFD) methods. The PFD methods just collect system input and output information for FD, while the AFD methods additionally design inputs to excite the system to obtain extra fault information. Therefore, the PFD methods are more simple but lead to more conservative FD results than the AFD methods. This paper considers two main challenges of set-based robust FD methods:
	\begin{itemize}
		\item Gain optimization design for set-based PFD;
		\item Joint gain and input design for set-based AFD.
	\end{itemize}
	
	The gain optimization design for a set-valued observer (SVO) takes into account that FD decisions are made by testing the consistency between measured and estimated system outputs. Since the bounds of uncertainties are considered, the SVO-based methods have the advantage being robust to uncertainties but provide relatively conservative FD performance. Therefore, the gain optimization of SVO is a key issue to keep the advantage of robustness of set-based FD methods and simultaneously reduce their conservatism. The gain optimization of observers for FD has been widely considered under the stochastic FD frameworks in the literature, where an important idea is to attenuate the effect of uncertainties on FD or maximize the relative magnitude of faults with respect to uncertainties \cite{Ding2013}. Particularly, this is because under the stochastic framework representation of uncertainties, state and output estimations and residuals are all vectors. In this case, it is straightforward to formulate optimization problems considering the trade-off between FD performance and robustness to design FD observer gains. However, when considering that uncertainties are bounded and that SVOs are used to implement robust FD, all estimations for states and outputs and residuals are in the form of sets. This results in difficulties on how to formulate the relative magnitude of faults with respect to uncertainties in the form of sets. Consequently, the gain optimization of SVO-based FD is still a relatively open topic. However, different from FD, the gain optimization of SVOs for state estimation is less challenging. This is because the optimality of SVO for state estimation only depends on the sizes of state estimation sets. With this idea, the gain optimization of SVO-based state estimation was done in \cite{Combastel2015265,Xu2019Auto} by designing observer gains to minimize the Frobenius norm ($F$-norm) sizes of the state estimation zonotopes. Besides, based on \cite{Combastel2015265}, time-varying locally optimal gains were designed in \cite{Xu2020RNC1} to optimize interval observers for state estimation as well.
	
	Within the authors' knowledge scope, there only exist few works to consider the effect of observer gains on FD of SVOs. Particularly, in \cite{MESEGUER2010JPC}, a systematic analysis of the effect of observer gains on fault detectability of SVO was carried out. In \cite{Xu2019JPC}, a locally optimal observer gain was computed  for FD by using a gain from \cite{Combastel2015265} at a certain time instant (after a sufficiently long time). Different from these two works, the method in \cite{PourasgharJPC2019} proposed a method to directly design gains to improve the performance of SVO-based FD, where FD gains were designed by maximizing the relative magnitude of faults to uncertainties based on a fractional objective function. However, in \cite{PourasgharJPC2019}, only generator matrices (i.e., sizes) of zonotopes of residuals and uncertainties were considered, which simplified the description of relative magnitude of faults to uncertainties under the set-based framework. Particularly, in the formulated fractional objective function, the numerator was the component of $F$-norm sizes of residual zonotopes related to faults while the denominator was the component of $F$-norm sizes of residual zonotopes related to uncertainties. Finally, FD gains were obtained based on the generalized eigenvalues at each time instant. Additionally, the method in \cite{PourasgharJPC2019} was extended to design FD gains for the set-theoretic unknown input observer of descriptor systems in \cite{WangJFI2019}. However, the method proposed in \cite{PourasgharJPC2019} has an important shortcoming since FD gains computed based on the generalized eigenvalues were not optimal for the fractional objective function and that the designed observer stability is not guaranteed. This is because the fractional objective function in \cite{PourasgharJPC2019} was constrained while the used method in \cite{PourasgharJPC2019} can only compute the optimal solution of the unconstrained case. 
	
	In order to address this problem, a fractional programming method was proposed in \cite{TanRNC2020} to obtain optimal FD gains corresponding to the fractional objective function in \cite{PourasgharJPC2019}. However, \cite{PourasgharJPC2019} and \cite{TanRNC2020} only considered the effect of generator matrices of residual zonotopes on set-based FD. In general, except for generator matrices, the centers of residual zonotopes also affect fault detectability of SVOs. Although \cite{PourasgharJPC2019} and \cite{TanRNC2020} considered time-varying and linear parameter-varying (LPV) systems, their proposed principles of designing optimal FD gains can be used for linear time-invariant (LTI) systems. Motivated by this fact, the authors' preliminary results \cite{Xu2022CDC} proposed a new notion named the excluding degree of the origin from a residual zonotope to simultaneously consider the effect of the centers and generator matrices of residual zonotopes to improve the FD performance. The advantage of the method in \cite{Xu2022CDC} consists in that the proposed excluding degree has an analytical expression according with the set-based FD criterion and can reflect the FD performance directly. However, the main disadvantage of the proposed excluding degree in \cite{Xu2022CDC} consists in that it cannot provide a specific value to the excluding degree in the case that the origin is just on the border of a residual zonotope. This means that the excluding degree value that guarantees fault detection is unknown. 
	
	In order to overcome this shortcoming, \cite{FanAuto2023} proposed a new exclusion tendency notion that characterizes the degree of residual zonotopes to exclude the origin based on a linear programming (LP) problem and is able to characterize the critical point of set-based FD. In particular, when the exclusion tendency of the origin from a residual zonotope takes a value larger than $1$, it guarantees that faults are detected. Otherwise, it is judged that the system is still healthy. However, the main disadvantage of the exclusion tendency in \cite{FanAuto2023} consists in that it relies on solving a linear programming (LP) problem and does not have an analytical expression. This means that both the excluding degree and the exclusion tendency can describe the set-based FD performance but both of them have their own advantages and disadvantages. Besides, when designing FD gains based on the excluding degree in \cite{Xu2022CDC}, the output $y_{k+1}$ was not used and thus resulted in some conservatism of SVO-based FD performance. It is also realized that all the above works on gain optimization of SVOs depend on solving optimization problems online to obtain FD gains, which implies high online computational complexity especially for large-scale systems. Motivated by this drawback, the latest work \cite{WangCR2024111376} proposed a notion named unguaranteed detectable faults set and then designed an analytical optimal gain by minimizing the Frobenius radius ($F$-radius) of unguaranteed detectable faults set, which achieves the lowest computational complexity among all the existing gain optimization methods of set-based FD. Note that all the works above reflect the state of the art of the gain optimization study of set-based FD.
	
	On the joint gain and input design for set-based AFD, it is still an open issue in the literature and its key challenge consists in the nonlinear couplings of gains and inputs, which makes the formulation and solution of AFD gain and input design problems difficult. In the literature, we can divide the existing set-based AFD methods into two approaches. The first approach generates state and output reachable sets based on the set-based versions of system models. The second approach generates state and output estimation sets based on SVOs. Within the authors' knowledge scope, the early results on set-based AFD belong to the first approach. In particular, the earliest work on set-based AFD was proposed in \cite{Nikoukhah1998}, where AFD inputs were designed to separate two polytopic output sets, i.e., a healthy set and a faulty set. In \cite{Tabatabaeipour2015IJSS}, AFD inputs were designed to separate a group of polytopic output sets. The polytopes are flexible to describe complex geometric shapes and are closed under the common set operations such as Minkowski sum, mapping by an appropriate matrix, etc. However, on one hand, polytopes results in high computational complexity in the design of AFD inputs. On the other hand, they are computationally unstable when designing AFD inputs to separate a group of polytopes especially for systems with dimensions larger than ten \cite{Scott20141580}. In order to overcome the computational unstability issue of polytopes, the work \cite{Scott20141580} proposed to use zonotopes to replace polytopes and designed AFD inputs by separating a group of output zonotopes. Due to the centrosymmetric feature of zonotopes, the computational unstability issue is overcome under the zonotopic framework. This is because the method in \cite{Scott20141580} transformed the separation constraints of a group of output zonotopes into a group of mixed-integer constraints. Consequently, a mixed-integer quadratic programming (MIQP) problem was formulated to design AFD inputs, which has exponentially increasing computational complexity with respect to the system dimension, the number of faults and AFD input sequence length. In order to reduce computational complexity, \cite{Marseglia2017Auto} proposed to use parametric programming methods to separate two zonotopes out of a group of zonotopes at a time. In our opinion, the root of computational complexity of the introduced works above is the set separation-based AFD input design conditions.
	
	Differently, the second approach turns to a bank of SVOs instead of set versions of a bank of system models. In \cite{Raimondo2016Auto}, the idea in \cite{Scott20141580} was further extended to closed-loop AFD input design based on set membership estimators. In \cite{Wang2023IJRNC}, an integrated AFD and control scheme was proposed, where AFD was still done based on the set membership estimators and the set separation conditions. In \cite{ZhangTASE2023}, an AFD input design method was proposed based on the separation of a group of zonotopes for linear time-varying (LPV) systems. Since AFD inputs are still designed based on the separation of a group of output zonotopes, the observer-based AFD methods proposed in \cite{Raimondo2016Auto,ZhangTASE2023,Wang2023IJRNC} still have high computational complexity. In order to overcome the computational complexity issue, \cite{Xu2021Auto1} proposed a new SVO-based AFD framework introducing a new notion named set separation tendency, and designed AFD inputs to increase the separation tendency of output zonotopes step by step such that AFD was finally achieved. In particular, the set separation tendency is increased by designing observer gains to minimize the size of all output sets and simultaneously designing inputs to increase the centers distance of all output sets. The simulation results show that the proposed method in \cite{Xu2021Auto1} has much less computational complexity than that of the output set separation-based methods \cite{Scott20141580,Raimondo2016Auto}. Since the observer gains designed in \cite{Xu2021Auto1} aim to minimize the size of output zonotopes instead of directly optimizing an FD objective, following the SVO-based AFD framework proposed in \cite{Xu2021Auto1}, the latest work \cite{FanAuto2023} further proposed a new FD gain design method to improve the FD performance of the whole SVO-based AFD framework. However, due to that the nonlinear couplings result in challenges for joint gain and input design, both \cite{FanAuto2023} and \cite{Xu2021Auto1} can only design gains and AFD inputs, separately, which means some FD performance losses. Besides, there are also some interesting AFD works considering stochastic and hybrid uncertainties. For example, the work \cite{Guo2024Auto} effectively addressed the input design problem incorporating chance constraints and presented the relationship between the misdiagnosis probability and the input signals for stochastic system, resulting in superior diagnosis performance and lower conservatism, which is a remarkable milestone for stochastic AFD. The work \cite{Scott2013CDC} considered that uncertainties were both bounded and subject to stochastic distributions and improved the AFD performance. For more additional AFD works on stochastic, set-based and hybrid AFD methods, the readers are referred to \cite{Heirung201935,TAN2021TAC,Cao9775733,Qiu2023Auto}.
	
	To exploit the potential of the SVO-based AFD framework and based on preliminary results in \cite{Xu2022CDC}, the main contributions of this paper are summarized as follows:
	\begin{itemize}
		\item Based on an analytical excluding degree specification, a new observer gain design method is proposed to optimize the FD performance;
		\item Based on the excluding degree, a novel joint observer gain and input design method is further proposed for the SVO-based AFD framework. 
	\end{itemize}
	
	The remainder of this paper is organized as follows. Section 2 introduces some preliminaries, the system model and the SVO. The results on gain and input optimization design are presented in Section 3. In Section 4, examples are used to illustrate the effectiveness of the proposed methods. The paper is concluded in Section 5.

	\section{Preliminaries and Problem Formulation}
	\label{Section2}
	\subsection{Preliminaries}
	\label{Section2_1}
	The identity and null matrices of appropriate dimensions are denoted by $I$ and $O$, respectively. \textbf{0} denotes an appropriate vector full of zeros. $\mathrm{diag}(\cdot)$ denotes a (block) diagonal matrix. $\Vert\cdot\Vert$ denotes the Euclidean norm. $\otimes$ denotes the Kronecker product. $\mathrm{vec}(\cdot)$ is the vectorizing operation that transforms a matrix into a column vector. $\mathrm{vec}^{-1}(\cdot)$ is an inverse operation of $\mathrm{vec}(\cdot)$. $\times$ denotes the Cartesian product. An interval $[\underline{\gamma},\overline{\gamma}]$ of a scalar $\gamma$ is a closed and connected subset of $\mathbb{R}$ ($\mathbb{R}$ is the set of real numbers). An interval matrix or vector is represented by bold letters and is a matrix or vector whose elements are intervals. For an interval matrix, $\mathrm{mid}(\cdot)$ and $\mathrm{rad}(\cdot)$ denote its middle point and radius, respectively. The Minkowski sum of two sets $X$ and $Y$ is $X\oplus Y=\{x+y\ |\ x\in X, y\in Y\}$. A zonotope $Z$ is defined as $Z=g \oplus H\mathbb{B}^{r}$, where $g\in \mathbb{R}^{n}$ and $H \in \mathbb{R}^{n \times r}$ are its center vector and generator matrix, respectively, and $\mathbb{B}^{r}$ is a box composed of $r$ unitary intervals. For brevity, $Z=g \oplus H\mathbb{B}^{r}$ is abbreviated as $Z=\langle g,H\rangle$. Given $Z_1=\langle g_1,H_1\rangle$ and $Z_2=\langle g_2,H_2\rangle$, $Z_1 \oplus Z_2 = \langle g_1+g_2,[H_1\:\: H_2]\rangle$.	Given $Z=\langle g,H\rangle$ and an appropriate matrix $K$, $KZ=\langle Kg,KH\rangle$. The $F$-radius of a zonotope $Z$ with generator matrix $H \in
	\mathbb{R}^{n \times r}$ is the Frobenius norm of $H$, i.e., $\|Z\|_{F}=\|H\|_{F}=\sqrt{\sum_{i=1}^r\|h^i\|^2}$, where $\|h^i\|=\sqrt{(h^i)^Th^i}$ and $h^i$ is the $i$-th column of $H$ \cite{Alamo2005}. The square of the Frobenius radius of a zonotope is called its $F$-norm size for brevity in this paper.
	

	\subsection{Problem Formulation}
	\label{Section2_2}
	In this paper, the discrete LTI system subject to multiplicative actuator faults is modeled as
	\begin{subequations}
		\label{PlantDynamics1}
		\begin{align}
			\label{PlantDynamics1A}
			x_{k+1}=&Ax_k+BG_iu_k+E\omega_k, \\
			\label{PlantDynamics1B}
			y_k=&Cx_k+F\eta_k,
		\end{align}
	\end{subequations}
	where $A \in \mathbb{R}^{n_x \times n_x}$, $B \in \mathbb{R}^{n_x \times n_u}$, $C \in \mathbb{R}^{n_y \times n_x}$, $E \in \mathbb{R}^{n_x \times n_{\omega}}$ and $F \in \mathbb{R}^{n_y \times n_{\eta}}$ are parametric matrices, $k$ denotes the $k$-th discrete time instant, $x_k \in \mathbb{R}^{n_x}$ and $y_k \in \mathbb{R}^{n_y}$ are the state and output vectors, respectively, $u_k \in \mathbb{R}^{n_u}$ is the input vector, $\omega_k \in \mathbb{R}^{n_{\omega}}$ represents the process disturbance vector, and $\eta_k \in \mathbb{R}^{n_{\eta}}$ is the measurement noise vector. Without loss of generality, this paper considers single faults. Consequently, there are $n_u+1$ actuator modes (the healthy mode and $n_u$ faulty modes) considered here and $G_i \in \mathbb{R}^{n_u \times n_u}$ ($i\in \mathbb{I}_u=\{0,1,2,\cdots,n_u\}$) is used to model the $i$-th actuator mode of the system \eqref{PlantDynamics1}. Particularly, $G_0$ is an identity matrix modeling the mode that all actuators are healthy and $G_i=\text{diag}([1,\cdots,1,g_i,1,\cdots,1])$ ($i\in \mathbb{I}_u\backslash\{0\}$) is a diagonal matrix modeling a fault in the $i$-th actuator, where $g_i=1$ means that the $i$-th actuator is healthy and $g_i \in [0,1)$ means that the $i$-th actuator is faulty.
	
	\begin{assumption}
		\label{SystemDetectability} The pair ($A$, $C$) in \eqref{PlantDynamics1} is detectable.
	\end{assumption}
	
	\begin{assumption}
		\label{Uncetainties} $\omega_k$ and
		$\eta_k$ are unknown but bounded by $\mathcal{W} = \langle \omega^c,H_{\omega} \rangle$ and
		$\mathcal{V} =\langle \eta^c,H_{\eta} \rangle $, respectively.
	\end{assumption}
	
	\assref{Uncetainties} is a classical assumption for set-based methods. In reality, it is difficult to know the value $G_i$ for a fault (i.e., $i\not=0$) and thus we consider a bound $\mathbf{G}_i$ of $G_i$ (i.e., $G_i \in \mathbf{G}_i$). As stated above, $G_i$ has its $i$-th diagonal element $g_i$ unknown while all its other diagonal elements are $1$. Thus, the bound $\mathbf{G}_i$ is obtained as $G_i\in \mathbf{G}_i=\text{diag}([1,\cdots,1,[0,1),\cdots,1])$ where the interval $[0,1)$ is located in the $i$-th diagonal position of $\mathbf{G}_i$.
	
	This paper designs a bank of SVOs for FD, where each SVO matches an actuator mode. Particularly, the $i$-th SVO for the $i$-th actuator mode of \eqref{PlantDynamics1} is designed as
	\begin{subequations}
		\label{Observer1}
		\begin{align}
			\label{Observer1A}
			\hat{\mathcal{X}}^i_{k+1}=&(A-L^i_kC)\hat{\mathcal{X}}^i_k \oplus B\mathbf{G}_iu_k \oplus L^i_ky_k \oplus E\mathcal{W} \notag\\&\oplus (-L^i_kF\mathcal{V}), \\
			\label{Observer1B}
			\hat{\mathcal{Y}}^i_{k}=&C\hat{\mathcal{X}}^i_{k} \oplus F\mathcal{V},
		\end{align}
	\end{subequations}
	where $\hat{\mathcal{X}}^i_{k}$ and $\hat{\mathcal{Y}}^i_{k}$ are the state and output sets, respectively, and $L^i_k$ is the gain. For $x_0 \in \hat{\mathcal{X}}^i_0$, if the system is in the $i$-th mode, $x_k \in \hat{\mathcal{X}}^i_k$ and $y_k \in \hat{\mathcal{Y}}^i_k$, $\forall k \ge 0$.
	
	When the system operates at the steady state of the $i$-th mode, the FD criterion is to test whether
	\begin{align}	
		\label{FD0}
		y_{k}\in \hat{\mathcal{Y}}^i_{k}
	\end{align}
	holds. If \eqref{FD0} is violated, it implies fault occurrence. Otherwise, no fault can be indicated.
	
	Besides, as $k$ increases, the order of $\hat{\mathcal{X}}^i_{k}$ also increases. In order to avoid computational explosion, the reduction method in \cite{Combastel2015265} is used to control the order of $\hat{\mathcal{X}}^i_{k}$. 
	
	The objective of this paper is to jointly design $L^i_k$ ($i\in \mathbb{I}_u$) and $u_k$ in \eqref{Observer1} for maximizing the FD performance. Consequently, the key challenge consists in proposing a new theoretic framework to characterize the joint gain and input design problem and formulate it into a solvable optimization problem. The details will be given below.

	\section{Main Results}
	\label{Section4}
	\subsection{Characterization of Fault Diagnosis Performance}
	\label{Section4_1}
	To handle the term $B\mathbf{G}_{i} u_{k}$ in \eqref{Observer1}, the following results are obtained in \lemref{intervalmatrix}.
	\begin{lemma}
		\label{intervalmatrix}
		For an interval matrix $\mathbf{G}_i$ ($i\not=0$), one has $B \mathbf{G}_{i} u_{k}\subseteq \langle \operatorname{mid}(B\mathbf{G}_i) u_{k}, \frac{1-\epsilon_1}{2}\operatorname{diag}\left(Be_i\right)u^i_k\rangle$ where $\epsilon_1$ is a small enough positive scalar, $e_i$ is an appropriate unit column vector whose $i$-th element is $1$ and other elements are zeros, and $u^i_k$ is the $i$-th element of $u_k$. 
	\end{lemma}
	\begin{proof}
		According to \cite{Xu2021Auto1}, it is known that $B \mathbf{G}_{i} u_{k}\subseteq \left\langle \operatorname{mid}(B\mathbf{G}_i) u_{k}, \operatorname{diag}\left(\operatorname{rad}(B\mathbf{G}_i) u_{k}\right)\right\rangle$. Considering $\mathbf{G}_i=\mathrm{diag}([1,\cdots,1,[0,1),\cdots,1])$, we could rewrite it as
		$\mathbf{G}_i=\mathrm{diag}([1,\cdots,1,[0,1-\epsilon_1],\cdots,1])$ and $\operatorname{rad}(B\mathbf{G}_i) u_{k}=\frac{1-\epsilon_1}{2}Be_iu^i_k$,
		which means 
		\begin{equation*}
			\operatorname{diag}\left(\operatorname{rad}(B\mathbf{G}_i) u_{k}\right)=\frac{1-\epsilon_1}{2}\operatorname{diag}\left(Be_i\right)u^i_k. \quad\quad \quad
		\end{equation*}
	\end{proof}	
	
	Based on \lemref{intervalmatrix}, and the Minkowski sum and linear mapping operations of zonotopes, \eqref{Observer1} can be decomposed into the center-generator matrix form:
	\begin{subequations}
		\label{Observer2}
		\begin{align}
			\label{ObserverA}
			\hat{x}^{i,c}_{k+1}=&(A-L^i_kC)\hat{x}^{i,c}_k + L^i_ky_k-L^i_kF\eta^c\notag\\&+\operatorname{mid}(B\mathbf{G}_i)u_k+ E\omega^c,\\
			\label{ObserverB}
			\hat{H}^{i,x}_{k+1}=&[(A-L^i_kC)\hat{H}^{i,x}_k\:\: -L^i_kFH_{\eta}\:\: \notag\\& \frac{1-\epsilon_1}{2}\operatorname{diag}\left(Be_i\right)u^i_k\:\:  EH_{\omega}],\\
			\label{ObserverC}
			\hat{y}^{i,c}_{k+1}=&C\hat{x}^{i,c}_{k+1}+F\eta^c,\\
			\label{ObserverD}
			\hat{H}^{i,y}_{k+1}=&[C\hat{H}^{i,x}_{k+1}\:\: FH_{\eta}],
		\end{align}
	\end{subequations}
	where $\hat{x}^{i,c}_{k+1}$ and $\hat{H}^{i,x}_{k+1}$, and $\hat{y}^{i,c}_{k+1}$ and $\hat{H}^{i,y}_{k+1}$ are the centers and generator matrices of $\hat{\mathcal{X}}^i_{k+1}$ and $\hat{\mathcal{Y}}^i_{k+1}$, respectively.
	
	It should be mentioned that when the system is healthy (i.e., $i=0$), \eqref{ObserverA} and \eqref{ObserverB} should be changed into
	\begin{subequations}
		\label{Observer2_1}
		\begin{align}
			\label{ObserverA1}
			\hat{x}^{i,c}_{k+1}=&(A-L^i_kC)\hat{x}^{i,c}_k + L^i_ky_k-L^i_kF\eta^c\notag\\&+Bu_k+ E\omega^c,\\
			\label{ObserverB1}
			\hat{H}^{i,x}_{k+1}=&[(A-L^i_kC)\hat{H}^{i,x}_k\:\: -L^i_kFH_{\eta}\:\: EH_{\omega}].
		\end{align}
	\end{subequations}
	
	When faults occur in a system, it is expected to timely detect them. Since this paper aims to design $L^i_k$ ($i\in \mathbb{I}_u$) and $u_k$ at time instant $k$ to speed FD up at the next time instant. Thus, the FD criterion \eqref{FD0} is rewritten as its form of the time instant $k+1$ for better expression and analysis in the following, i.e.,  
	\begin{align}
		\label{FD1}
		y_{k+1}\in \hat{\mathcal{Y}}^i_{k+1}
	\end{align}
	
	For better analysis, \eqref{FD1} is equivalently transformed into
	\begin{align}
		\label{FD2}
		\mathbf{0} \in \mathcal{R}^i_{k+1}=y_{k+1}\oplus (-\hat{\mathcal{Y}}^i_{k+1}),
	\end{align}
	where $\mathcal{R}^i_{k+1}$ is called the residual zonotope for the $i$-th SVO. For brevity, $\mathcal{R}^i_{k+1}$ is rewritten into $\mathcal{R}^i_{k+1}=\langle r^{i,c}_{k+1},H^{i,r}_{k+1} \rangle$,
	where $r^{i,c}_{k+1}$ and $H^{i,r}_{k+1}$ are the center and generator matrix, respectively, and have the following expressions:
	\begin{subequations}
		\label{ResidualZonotopeCenter}
		\begin{align}
			\label{ResidualZonotopeCenterA}
			r^{i,c}_{k+1}=&y_{k+1}-\hat{y}^{i,c}_{k+1}\notag\\=&y_{k+1}-C(A-L^i_kC)\hat{x}^{i,c}_k - C\operatorname{mid}(B\mathbf{G}_i)u_k \notag\\&- CL^i_ky_k - CE\omega^c + CL^i_kF\eta^c-F\eta^c,\\
			\label{ResidualZonotopeCenterB}
			H^{i,r}_{k+1}=&\hat{H}^{i,y}_{k+1}\notag\\=&[C(A-L^i_kC)\hat{H}^{i,x}_k \:\: \frac{1-\epsilon_1}{2}C\operatorname{diag}\left(Be_i\right)u^i_k\:\:CEH_{\omega}\notag\\&-CL^i_kFH_{\eta}\:\:FH_{\eta}].
		\end{align}
	\end{subequations}
	
	As stated in Section \ref{Section1}, the key challenge consists in how to describe the relative magnitude of faults to uncertainties in the form of sets considering their centers and generator matrices. Motivated by this fact and based on the preliminary results in \cite{Xu2022CDC}, this paper proposes a new specification to characterize the set-based FD performance. In particular, inspired by \eqref{FD2}, if the origin can be quickly excluded from $\mathcal{R}^i_{k+1}$, quick FD can be facilitated at time instant $k+1$. Following this idea, a new specification is proposed to describe the degree of excluding the origin from a zonotope (also called the excluding degree in this paper).
	\begin{definition}
		\label{DegreeofExcluding}
		The excluding degree of the origin from a zonotope $\mathcal{Z}=\langle g,H\rangle$ is defined as $\mathfrak{D}(\mathcal{Z})=\frac{\Vert g\Vert^2_2}{\Vert H\Vert^2_F}$
		where $\Vert g \Vert^2_2$ describe the distance\footnote{The Euclidean distance between the center of a zonotope $\mathcal{Z}=\langle g,H\rangle$ and the origin is $\Vert g \Vert_2$. However, for convenience to formulate and solve the optimization problems in the following, $\Vert g \Vert^2_2$ is used instead of $\Vert g \Vert_2$ in \defref{DegreeofExcluding}.} from the center of $\mathcal{Z}$ to the origin and $\Vert H\Vert^2_F$ measures the size of $\mathcal{Z}$.
	\end{definition}
	
	Under \defref{DegreeofExcluding}, at time instant $k+1$, the excluding degree of the origin from $\mathcal{R}^i_{k+1}$ is denoted by
	\begin{align}
		\label{FD_3}
		\mathfrak{D}(\mathcal{R}^i_{k+1})=\frac{\Vert r^{i,c}_{k+1}\Vert^2_2}{\Vert H^{i,r}_{k+1}\Vert^2_F}.
	\end{align}
	
	Generally, the larger $\mathfrak{D}(\mathcal{R}^i_{k+1})$ is, the easier the violation of \eqref{FD2} is. Considering this fact, if $\mathfrak{D}(\mathcal{R}^i_{k+1})$ in \eqref{FD_3} is maximized, the FD performance is optimized.

	\subsection{Gain Design for PFD}
	\label{Section4_2}
	Under the PFD framework, inputs are known signals for the FD module and thus only a parameter (i.e., the observer gain) is left. Based on the excluding degree, the optimal FD gain for the $i$-th SVO can be designed at time instant $k$ by maximizing $\mathfrak{D}(\mathcal{R}^i_{k+1})$. However, as observed in \eqref{ResidualZonotopeCenterA}, the computation of $\mathcal{R}^i_{k+1}$ is related to $y_{k+1}$. Therefore, when we design the optimal gain at time instant $k$, $y_{k+1}$ should be known such that only a design parameter (i.e., the gain) is left. Based on this analysis and motivated by \cite{FanAuto2023}, a different logic from the preliminary results proposed in \cite{Xu2022CDC} to design the optimal gain in this subsection is given as follows:
	\begin{itemize}
		\item At time instant $k$, an input $u_k$ is first injected into the system and then the output $y_{k+1}$ can be measured;
		\item Since $y_{k+1}$ is known, $\mathfrak{D}(\mathcal{R}^i_{k+1})$ has only one parameter (i.e., the gain). Then,  at time instant $k+1$, an optimal value of $L^i_k$ can be designed by maximizing $\mathfrak{D}(\mathcal{R}^i_{k+1})$.
	\end{itemize}
	
	Since the output information is employed in the logic above, the designed gains here have potential to achieve a better FD performance than those proposed in the preliminary results \cite{Xu2022CDC}. Based on this idea and using \eqref{ResidualZonotopeCenter} and \eqref{FD_3}, an optimal gain at time instant $k$ is designed by $\max_{L^i_k} \:\mathfrak{D}(\mathcal{R}^i_{k+1})$
	which is equivalent to
	\begin{align}
		\label{OptimalGain0}
		\min_{L^i_k} \:\frac{1}{\mathfrak{D}(\mathcal{R}^i_{k+1})}.
	\end{align}

	\begin{proposition}
		The optimization problem \eqref{OptimalGain0} is equivalent to a quadratic fractional programming problem
		\begin{align}
			\label{OptimalGain1}
			\min_{\xi^i_k} \:\frac{J_1(\xi^i_k)}{J_2(\xi^i_k)},
		\end{align}
		where
		\begin{subequations}
			\label{OptimalGain2}
			\begin{align}
				\xi^i_k=&\mathrm{vec}(L^i_k), \\
				J_1(\xi^i_k)=&(\xi^i_k)^TP^{i,1}_{k}\xi^i_k+P^{i,2}_{k}\xi^i_k+P^{i,3}_{k}, \\
				J_2(\xi^i_k)=&(\xi^i_k)^TQ^{i,1}_{k}\xi^i_k+Q^{i,2}_{k}\xi^i_k+Q^{i,3}_{k},\\
				P^{i,1}_{k}=&(C\hat{H}^{i,x}_k(\hat{H}^{i,x}_k)^TC^T +FH_{\eta}H^T_{\eta}F^T) \notag\\& \otimes (C^TC),\\
				P^{i,2}_{k}=&-2\mathrm{vec}\Big(C^TCA\hat{H}^{i,x}_k(\hat{H}^{i,x}_k)^TC^T\Big)^T,\\
				P^{i,3}_{k}=&\mathrm{tr}\Big(CA\hat{H}^{i,x}_k(\hat{H}^{i,x}_k)^TA^TC^T + CEH_{\omega}(CEH_{\omega})^T \notag\\&+ \frac{(1-\epsilon_1)^2}{4}C\operatorname{diag}\left(Be_i\right)u^i_k\big(C\operatorname{diag}\left(Be_i\right)u^i_k\big)^T \notag\\&+FH_{\eta}H^T_{\eta}F^T\Big) , \\
				Q^{i,1}_{k}=&(\beta\beta^T) \otimes (C^TC), Q^{i,2}_{k}=2\mathrm{vec}(C^T\alpha\beta^T)^T, \\
				Q^{i,3}_{k}=&\alpha^T\alpha, \beta=C\hat{x}^{i,c}_k - y_k + F\eta^c, \\
				\alpha=&y_{k+1}-CA\hat{x}^{i,c}_k - C\operatorname{mid}(B\mathbf{G}_i)u_k \notag\\&- CE\omega^c -F\eta^c.
			\end{align}
		\end{subequations}
	\end{proposition}
	\begin{proof}
		By substituting \eqref{ResidualZonotopeCenterA} and \eqref{ResidualZonotopeCenterB} into the numerator and denominator of \eqref{FD_3}, respectively, the explicit expression of $\mathfrak{D}(\mathcal{R}^i_{k+1})$ can be obtained. Consequently, the explicit form \eqref{OptimalGain1} and \eqref{OptimalGain2} of \eqref{OptimalGain0} is obtained. \quad\quad
	\end{proof}
	
	Therefore, by solving \eqref{OptimalGain1} at each time instant, a series of optimal $\xi^{i,*}_k$ for $k \ge 0$ are obtained and then optimal FD gains for the $i$-th SVO can be computed by using $L^{i,*}_k=\mathrm{vec}^{-1}(\xi^{i,*}_k),\:i\in \mathbb{I}_u$.
	Using the proposed method above, the optimal gain for each SVO at each time instant is independent of each other. Consequently, optimal FD gains for a bank of SVOs can be designed at each step in the same way.
	
	\begin{remark}
		The proposed gain design method above can also be applied for AFD. In particular, the application includes two steps. The first step is to design inputs to facilitate AFD \cite{FanAuto2023,Xu2021Auto1}. After the inputs are designed, the second step is to follow the procedure in this subsection to design optimal gains to further speed up FD.
	\end{remark}
	
	\subsection{Joint Gain and Input Design for AFD}
	\label{Section4_3}
	Different from PFD, AFD designs inputs to excite the system to obtain more fault information for diagnosis. Different from \cite{FanAuto2023,Xu2021Auto1} that design gains and inputs by two steps, the best way is to simultaneously optimize gains and inputs such that the potential of observer-based AFD is exploited. In this case, $y_{k+1}$ is unavailable at time instant $k$ and cannot be obtained in the same way used in Section \ref{Section4_2} as well. Since the computation of $\mathcal{R}^i_{k+1}$ depends on $y_{k+1}$, a different method should be further proposed to handle this difficulty such that a joint design of gains and inputs is finally implemented.
	
	Particularly, based on the system model, an additional set-based dynamics is designed for the $i$-th mode as
	\begin{align} 
		\label{SetModel}
		\mathcal{X}^{i}_{k+1} &= A\mathcal{X}^{i}_{k} \oplus B\Upsilon_i\mathbf{G}_iu_k \oplus E\mathcal{W},
	\end{align}
	where $\mathcal{X}^{i}_{k+1}$ is the state set and $\Upsilon_i=\text{diag}([1\cdots1\:\:\:(1+\epsilon_2)\:\:\:1\cdots1])$ with $i\not=0$, $\epsilon_2$ is a small positive scalar and $(1+\epsilon_2)$ is in the $i$-th diagonal position of $\Upsilon_i$. Note that when $i=0$, both $\Upsilon_0$ and $\mathbf{G}_0$ should be the appropriate identity matrices. Besides, the motivation of introducing $\Upsilon_i$ for the cases $i\not=0$ will be detailed in the following. 
	
	Thus, with \eqref{Observer1A} and \eqref{SetModel}, for each mode $i$, we have two set-based estimators. When the system is in the $i$-th mode, if $x_k \in \hat{\mathcal{X}}^{i}_{k}$ and $x_k \in \mathcal{X}^{i}_{k}$, then we have
	\begin{subequations}
		\label{SetModel2}
		\begin{align} 
			x_{k+1} &\in \hat{\mathcal{X}}^{i}_{k+1},\\
			x_{k+1} &\in \mathcal{X}^{i}_{k+1}.
		\end{align}
	\end{subequations}
	
	At time instant $k$, the real values of the $i$-th fault, disturbances and noises are denoted as $G_i$, $\omega_k$ and $\eta_k$ with $G_i \in \mathbf{G}_i$, $\omega_k \in \mathcal{W}$ and $\eta_k \in \mathcal{V}$, respectively. This means that we could further obtain the set-based dynamics: 
	\begin{subequations}
		\label{SetModel3}
		\begin{align}
			\label{SetModel3A}
			\bar{\hat{\mathcal{X}}}^i_{k+1}=&(A-L^i_kC)\hat{\mathcal{X}}^i_k \oplus BG_iu_k \oplus L^i_ky_k \oplus E\omega_k \notag\\&\oplus (-L^i_kF\eta_k), \\
			\label{SetModel3B}
			\tilde{\mathcal{X}}^{i}_{k+1} &= A\mathcal{X}^{i}_{k} \oplus B\Upsilon_iG_iu_k \oplus E\omega_k.
		\end{align}
	\end{subequations}
	
	We always have $x_{k+1} \in \bar{\hat{\mathcal{X}}}^i_{k+1}$. However, due to $\Upsilon_i$, $x_{k+1} \in \tilde{\mathcal{X}}^{i}_{k+1}$ cannot be guaranteed. In this case, a new set-based dynamics is further proposed to replace \eqref{SetModel3B}:
	\begin{align}
		\label{SetModel4}
		\bar{\mathcal{X}}^{i}_{k+1} &= A\mathcal{X}^{i}_{k} \oplus BG_iu_k \oplus E\omega_k \oplus B(\Upsilon_i-\mathrm{I})\mathbf{G}_iu_k,
	\end{align}
	where if $x_k \in \mathcal{X}^{i}_{k}$, $x_{k+1} \in \bar{\mathcal{X}}^{i}_{k+1}$ can be guaranteed.
	
	\begin{remark}
		If $x_k \in \mathcal{X}^{i}_{k}$, then $x_{k+1} \in A\mathcal{X}^{i}_{k} \oplus BG_iu_k \oplus E\omega_k$. Since $\mathbf{0} \in B(\Upsilon_i-\mathrm{I})\mathbf{G}_iu_k$, we have $x_{k+1} \in \bar{\mathcal{X}}^{i}_{k+1}$. 
	\end{remark}
	
	Based on the analysis above, when the system is in the $i$-th mode, we always have
	\begin{align}
		\label{SetModel5}
		\bar{\mathcal{X}}^{i}_{k+1} \cap \bar{\hat{\mathcal{X}}}^i_{k+1} \not= \emptyset,
	\end{align}	
	which further implies
	\begin{align}
		\label{SetModel6}
		\mathbf{0} \in& \bar{\mathcal{X}}^{i}_{k+1} \oplus (-\bar{\hat{\mathcal{X}}}^i_{k+1})
		\notag\\=&A\mathcal{X}^{i}_{k} \oplus B(\Upsilon_i-\mathrm{I})\mathbf{G}_iu_k \oplus (-(A-L^i_kC)\hat{\mathcal{X}}^i_k) \notag\\&\oplus (-L^i_ky_k) \oplus L^i_kF\eta_k \subset \bar{\mathcal{E}}^i_{k+1} 
	\end{align}	
	with 
	\begin{align}
		\bar{\mathcal{E}}^i_{k+1}=&A\mathcal{X}^{i}_{k} \oplus B(\Upsilon_i-\mathrm{I})\mathbf{G}_iu_k \oplus (-(A-L^i_kC)\hat{\mathcal{X}}^i_k) \notag\\&\oplus (-L^i_ky_k) \oplus L^i_kFV.
	\end{align}
	
	Based on the properties of zonotopes and \lemref{intervalmatrix}, the center and generator matrix of $\bar{\mathcal{E}}^i_{k+1}$ are derived as
	\begin{subequations}
		\label{SetModel7}
		\begin{align} 
			\label{SetModel7A}
			\bar{e}_{k+1}^{i,c} =& Ax_{k}^{i,c}-(A-L_k^iC)\hat x_{k}^{i,c}+\frac{(1-\epsilon_1)\epsilon_2}{2}Be_iu^i_k \notag\\&-L_{k}^{i} y_{k}+L_{k}^{i}F \eta^{c},\\
			\label{SetModel7B}
			\bar{H}_{k+1}^{i,e} =& [ AH_{k}^{i,x} \quad -(A-L_k^iC)\hat{H}_{k}^{i,x} \quad L_k^iFH_{\eta}\notag\\&  \frac{(1-\epsilon_1)\epsilon_2}{2}\mathrm{diag}(Be_i)u^i_k],
		\end{align}
	\end{subequations}
	where $\bar{\mathcal{E}}^i_{k+1}=\langle\bar{e}_{k+1}^{i,c}, \bar{H}_{k+1}^{i,e}\rangle$ and $\mathcal{X}^i_{k}=\langle x_{k}^{i,c}, H_{k}^{i,x}\rangle$. 
	
	When $i=0$, we have $\bar{\mathcal{E}}^0_{k+1}=\langle\bar{e}_{k+1}^{0,c}, \bar{H}_{k+1}^{0,e}\rangle$ with
	\begin{subequations}
		\label{SetModel7_2}
		\begin{align} 
			\label{SetModel7_2A}
			\bar{e}_{k+1}^{0,c} =& A\bar{x}_{k}^{0,c}-(A-L_k^0C)\hat x_{k}^{0,c}-L_{k}^{0} y_{k}+L_{k}^{0}F \eta^{c},\\
			\label{SetModel7_2B}
			\bar{H}_{k+1}^{0,e} =& [ A\bar{H}_{k}^{0,x} \quad -(A-L_k^0C)\hat{H}_{k}^{0,x} \quad L_k^0FH_{\eta}].
		\end{align}
	\end{subequations}
	
	Based on the analysis above, if \eqref{SetModel6} is violated, i.e., 
	\begin{align}
		\label{SetModel8}
		\mathbf{0} \not\in \bar{\mathcal{E}}^i_{k+1},
	\end{align}	
	it is guaranteed that the system changes its mode from the $i$-th one to another one. This implies that if gains and inputs are designed to increase the excluding degree of the origin from $\bar{\mathcal{E}}^i_{k+1}$, FD can be speeded up. 
	
	According to \defref{DegreeofExcluding}, the excluding degree of the origin from $\bar{\mathcal{E}}^i_{k+1}$ is obtained as
	\begin{align}
		\label{SetModel9}
		\mathfrak{D}(\bar{\mathcal{E}}^i_{k+1})=\frac{\Vert \bar{e}^{i,c}_{k+1}\Vert^2_2}{\Vert \bar{H}^{i,e}_{k+1}\Vert^2_F}.
	\end{align}
	
	\begin{remark}
		If $\Upsilon_i$ is changed into $I$ in \eqref{SetModel}, the term $B\Upsilon_i\mathbf{G}_iu_k$ in \eqref{SetModel3B}is changed into $B\mathbf{G}_iu_k$, which implies the disappearance of $B(\Upsilon_i-\mathrm{I})\mathbf{G}_iu_k$ in \eqref{SetModel4}. This results in that the terms related to $u_k$ disappear in \eqref{SetModel7} and \eqref{SetModel9} such that joint gain and input design becomes impossible.
	\end{remark}
	
	Similarly, the idea becomes to optimize gains and inputs by maximizing $\mathfrak{D}(\bar{\mathcal{E}}^i_{k+1})$. However, due to the consideration of inputs, a joint design of gains and inputs for each SVO is not independent from each other in this case, which is the key difference from the PFD case in Section \ref{Section4_2}. Thus, when designing gains and inputs, the effect of all the considered modes should be considered. Motivated by this fact and based on \eqref{SetModel9}, a total excluding degree for all the considered modes is further defined as 
	\begin{align}
		\label{FD5_1_0}
		\mathfrak{D}(\bar{\mathcal{E}}^0_{k+1},\bar{\mathcal{E}}^1_{k+1},\cdots,\bar{\mathcal{E}}^{n_u}_{k+1})=\frac{\sum^{n_u}_{i=0}\sigma^i_k\Vert \bar{e}^{i,c}_{k+1}\Vert^2_2}{\sum^{n_u}_{i=0}\sigma^i_k\Vert \bar{H}^{i,e}_{k+1}\Vert^2_F},
	\end{align}
	where $\sigma^i_k$ is a weighting coefficient used to measure the relative importance of the $i$-th mode among all the considered modes. Since at time instant $k$, $\bar{\mathcal{E}}^i_{k}$ ($\forall i \in \mathbb{I}_u$) are known, $\mathfrak{D}(\bar{\mathcal{E}}^i_k)$ $(\forall i \in \mathbb{I}_u)$ can be computed for all the considered modes. Logically, for the $i$-th mode, the larger $\mathfrak{D}(\bar{\mathcal{E}}^i_k)$ is, the higher the possibility for the $i$-th mode not matching the real system mode is, and the larger a weight should be given to the $i$-th mode in \eqref{FD5_1_0} such that it can be quickly excluded. Based on the analysis here and the excluding degrees of all the modes at time instant $k$, $\sigma^i_k$ is defined as follows:
	\begin{align}
		\label{FD5_1}
		\sigma^i_k=\frac{\mathfrak{D}(\bar{\mathcal{E}}^i_k)}{\sum^{n_u}_{i=0}\mathfrak{D}(\bar{\mathcal{E}}^i_k)},\:0 \le \sigma^i_k \le 1.
	\end{align}
	
	Similar to \eqref{OptimalGain0}, optimal gains and inputs can be designed for AFD by solving the following optimization problem:
	\begin{align}
		\label{OptimalGainInput}
		\min_{L^0_k,L^1_k,\cdots,L^{n_u}_k,u_k} \:\frac{1}{\mathfrak{D}(\bar{\mathcal{E}}^0_{k+1},\bar{\mathcal{E}}^1_{k+1},\cdots,\bar{\mathcal{E}}^{n_u}_{k+1})}.
	\end{align}
	
	This subsection aims to implement a joint design of gains and inputs. The proposed idea for the joint design is to define a new variable to include all gains and inputs:
	\begin{align}
		\label{AugmentedVariable}
		\mu_k=[(\mathrm{vec}(L^0_k))^T,\cdots,(\mathrm{vec}(L^{n_u}_k))^T\:\:u^T_k]^T.
	\end{align}
	
	\begin{proposition}
		The optimization problem \eqref{OptimalGainInput} is equivalent to a quadratic fractional programming problem	
		\begin{align}
			\label{FD6}
			\min_{\mu_k} \:\frac{J_3(\mu_k)}{J_4(\mu_k)},
		\end{align}
		where
		\begin{subequations}
			\begin{align*}
				J_3(\mu_k)=&(\mu_k)^TP^4_k\mu_k+P^5_k\mu_k+P^6_k,\\
				J_4(\mu_k)=&(\mu_k)^TQ^4_k\mu_k+Q^5_k\mu_k+Q^6_k,\\
				P^4_k=&\begin{bmatrix} P^4_{11,k} & O \\O & P^4_{22,k}\end{bmatrix}, Q^4_k=\begin{bmatrix} Q^4_{11,k} & Q^4_{12,k} \\Q^4_{21,k} & Q^4_{22,k}\end{bmatrix},\\
				P^4_{11,k}=&\mathrm{diag}(\begin{bmatrix}\Gamma^{0,1}_k &\Gamma^{1,1}_k &\cdots&\Gamma^{n_u,1}_k\end{bmatrix}), \\
				P^4_{22,k}=&\mathrm{diag}(\begin{bmatrix}\Gamma^{1,2}_k &\cdots&\Gamma^{n_u,2}_k\end{bmatrix}),\\
				P^5_k=&\begin{bmatrix}\Gamma^{0,3}_k&\Gamma^{1,3}_k&\cdots&\Gamma^{n_u,3}_k & O\end{bmatrix},\\
				P^6_k=&\sum^{n_u}_{i=0}\sigma^i_k\mathrm{tr}\big(A\hat{H}^{i,x}_k(A\hat{H}^{i,x}_k)^T  +AH^{i,x}_k(AH^{i,x}_k)^T\big),\\
				Q^4_{11,k}=&\mathrm{diag}(\begin{bmatrix}\Gamma^{0,4}_k &\Gamma^{1,4}_k &\cdots&\Gamma^{n_u,4}_k\end{bmatrix}),\\
				Q^4_{21,k}=&(Q^4_{12,k})^T=\mathrm{diag}(\begin{bmatrix}\Gamma^{1,5}_k &\Gamma^{2,5}_k &\cdots&\Gamma^{n_u,5}_k\end{bmatrix}), \\
				Q^4_{22,k}=&\mathrm{diag}(\begin{bmatrix}\Gamma^{1,6}_k &\Gamma^{2,6}_k &\cdots&\Gamma^{n_u,6}_k\end{bmatrix}), \\
				Q^5_k=&\begin{bmatrix}\Gamma^{0,7}_k &\cdots&\Gamma^{n_u,7}_k & \Gamma^{1,8}_k&\cdots&\Gamma^{n_u,8}_k\end{bmatrix},\\
				Q^6_k=&\sum^{n_u}_{i=0}\sigma^i_k(x^{i,c}-\hat{x}^{i,c}_k)^TA^TA(x^{i,c}-\hat{x}^{i,c}_k),\\
				\Gamma^{i,1}_k=&\sigma^i_k\Big(\big(C\hat{H}^{i,x}_k(C\hat{H}^{i,x}_k)^T+FH_{\eta}(FH_{\eta})^T\big)\otimes I\Big),\\
				\Gamma^{i,2}_k=&\sigma^i_k\Big(I\otimes\big(\frac{(1-\epsilon_1)^2\epsilon_2^2}{4}\operatorname{diag}\left(Be_i\right)\operatorname{diag}\left(Be_i\right)\big)\Big),\\
				\Gamma^{i,3}_k=&-2\sigma^i_k\Big(\mathrm{vec}\big(A\hat{H}^{i,x}_k(\hat{H}^{i,x}_k)^TC^T\big)\Big)^T,\\
				\Gamma^{i,4}_k=&\sigma^i_k\big((C\hat{x}^{i,c}_k + F\eta^c-y_k)(C\hat{x}^{i,c}_k + F\eta^c- y_k)^T \otimes I \big),\\
				\Gamma^{i,5}_k=&\sigma^i_k\frac{(1-\epsilon_1)\epsilon_2}{2}\big((C\hat{x}^{i,c}_k + F\eta^c- y_k)^T\otimes e^T_iB^T\big),\\
				\Gamma^{i,6}_k=&\frac{\sigma^i_k(1-\epsilon_1)^2\epsilon^2_2}{4}e^T_iB^TBe_i,\\
				\Gamma^{i,7}_k=&2\sigma^i_k(C\hat{x}^{i,c}_k + F\eta^c-y_k)^T\big(I\otimes(x^{i,c}-\hat{x}^{i,c}_k)^TA^T\big),\\
				\Gamma^{i,8}_k=&\sigma^i_k(1-\epsilon_1)\epsilon_2(x^{i,c}-\hat{x}^{i,c}_k)^TA^TBe_i, i\in \mathbb{I}_u.
			\end{align*}
		\end{subequations}
	\end{proposition}
	\begin{proof}
		By substituting \eqref{SetModel7A}, \eqref{SetModel7B} and \eqref{FD5_1} into the numerator and denominator of \eqref{FD5_1_0}, respectively, we can obtain the explicit expression of $\mathfrak{D}(\bar{\mathcal{E}}^0_{k+1},\bar{\mathcal{E}}^1_{k+1},\cdots,\bar{\mathcal{E}}^{n_u}_{k+1})$ with respect to $L_k^i$ and $u_k$. Furthermore, using \eqref{AugmentedVariable}, $\mathfrak{D}(\bar{\mathcal{E}}^0_{k+1},\bar{\mathcal{E}}^1_{k+1},\cdots,\bar{\mathcal{E}}^{n_u}_{k+1})$ is reformulated to an equivalent form with respect to $\mu_k$ such that \eqref{FD6} is obtained. \quad\quad \quad
	\end{proof}	
	
	Similarly, by solving \eqref{FD6}, an optimal $\mu^{*}_k$ can be obtained. Optimal gains $L^{*}_k$ and input $u^*_k$ can be computed for observer-based AFD based on $\mu^{*}_k$ and \eqref{AugmentedVariable}. By optimizing gains and input using the proposed method above and then injecting to the system at each time instant, FD decisions can be done by testing \eqref{FD1} or \eqref{FD2} online. 
	
	\subsection{Solution of Unconstrained Problem}
	\label{Section4_4}
	First, we consider the case that all gains and inputs are unconstrained. If the designed SVOs are stable, we will be able to see the FD potential of the proposed method. Since \eqref{OptimalGain1} and \eqref{FD6} belong to the same type of optimization problems. Without loss of generality, only \eqref{FD6} is considered in this subsection for saving space and the methods to solve \eqref{FD6} can be applied to \eqref{OptimalGain1} directly.
	
	In the unconstrained case, we consider $\mu_k \in \mathcal{S}$ for brevity of analysis where $\mathcal{S}$ is a convex compact set, i.e.,
	\begin{align}
		\label{OptimalGain3}
		\min_{\mu_k \in \mathcal{S}} \:\frac{J_3(\mu_k)}{J_4(\mu_k)}.
	\end{align}
	
	When $\mathcal{S}$ is sufficiently large, \eqref{OptimalGain3} is equivalent to \eqref{FD6}. Under this condition, a method is further proposed to solve \eqref{FD6} based on \eqref{OptimalGain3}. To solve \eqref{OptimalGain3}, it is reformulated as a parametric programming problem with the following formulation:
	\begin{align}
		\label{OptimizationProblem8B}
		M(\gamma)&=\min_{\mu_k \in \mathcal{S}}J_3(\mu_k)-\gamma J_4(\mu_k),\: \gamma \in \mathbb{R}.
	\end{align}
	
	Based on some results in \cite{Dinkelbach1967}, we could obtain some properties of $M(\gamma)$ in \eqref{OptimizationProblem8B} as follows:
	\begin{itemize}
		\item $M(\gamma)$ is concave for $\gamma \in \mathbb{R}$;
		\item $M(\gamma)$ is strictly monotonic decreasing for $\gamma \in \mathbb{R}$;
		\item $M(\gamma)=0$ has a unique solution for $\gamma \in \mathbb{R}$.
	\end{itemize}
	
	\begin{lemma}
		\label{proposition1}
		(\cite{Dinkelbach1967}). The two problems are equivalent:
		\begin{align*}
			\text{I:}&\min_{\mu_k \in \mathcal{S}}\frac{J_3(\mu_k)}{J_4(\mu_k)}=\gamma;
			\:\text{II:}\min_{\mu_k \in \mathcal{S}}J_3(\mu_k)-\gamma J_4(\mu_k)=0.
		\end{align*}
	\end{lemma}
	
	Using \lemref{proposition1}, the optimization problem \eqref{OptimalGain3} can be transformed into finding a value $\gamma^*_k$ such that $M(\gamma^*_k)=0$. Then, the optimal solution $\mu^*_k$ of \eqref{OptimalGain3} is obtained as $\mu^{*}_k=\arg\min_{\mu_k \in \mathcal{S}}J_3(\mu_k)-\gamma^*_k J_4(\mu_k)$.
	By using \eqref{FD6}, \eqref{OptimizationProblem8B} is further transformed into
	\begin{align}
		\label{OptimizationProblem8_2}
		M(\gamma)=\min_{\mu_k \in \mathcal{S}}\:\mu_k^T\bar{O}^1_k(\gamma)\mu_k + \bar{O}^2_k(\gamma)\mu_k + \bar{O}^3_k(\gamma),
	\end{align}	
	where $\bar{O}^1_k(\gamma)=P^4_k-\gamma Q^4_k$, $\bar{O}^2_k(\gamma)=P^5_k-\gamma Q^5_k$ and $\bar{O}^3_k(\gamma)=P^6_k-\gamma Q^6_k$.
	
	Since $M(\gamma)$ is monotonically decreasing, $M(\gamma)=0$ has a unique solution $\gamma^*_k$ at time instant $k$. Thus, a bisection method is proposed in \algref{OptimalFD} to search $\gamma^*_k$ and then the key point is to obtain a bound of $\gamma^*_k$ (i.e., $\underline{\gamma} \le \gamma^*_k \le \overline{\gamma}$) for the bisection method \cite{TanRNC2020}. Since $\mathcal{S}$ is a convex compact set, \eqref{OptimalGain3} must have a minimum. Moreover, \eqref{OptimalGain3} with a sufficiently large $\mathcal{S}$ is equivalent to \eqref{FD6}. Since $\mathcal{S}$ is sufficiently large, \eqref{OptimizationProblem8_2} has a minimum only when $\bar{O}^1_k$ is non-negative definite. Thus, the upper bound $\overline{\gamma}$ of $\gamma$ can be obtained by solving
	\begin{align}
		\label{OptimizationProblem8_Gamma_1}
		\overline{\gamma}=\arg\max \:\gamma,\:\mathrm{s.t.}\:\bar{O}^1_k(\gamma)\succeq 0.
	\end{align}
	
	As known from \defref{DegreeofExcluding}, both the numerator and denominator of the excluding degree are positive. Therefore, based on \lemref{proposition1}, $J_3(\mu_k) >0$, $J_4(\mu_k) >0$ and $\gamma >0$. This means that the lower bound of $\gamma$ is $\underline{\gamma}=0$. After obtaining the interval $\gamma\in [\underline{\gamma},\overline{\gamma}]$, for any given value $\gamma\in [\underline{\gamma},\overline{\gamma}]$, $\bar{O}^1_k(\gamma)$ is positive semidefinite and \eqref{OptimizationProblem8_2} is unconstrained and convex . Thus, the analytical solution of \eqref{OptimizationProblem8_2} can be obtained:
	\begin{align}
		\label{OptimalSolution}
		\mu^{\gamma}_k=-\frac{1}{2}(\bar{O}^1_k(\gamma))^{-1}(\bar{O}^2_k(\gamma))^T
	\end{align}	
	by solving $\frac{\partial \big(\mu_k^T\bar{O}^1_k(\gamma)\mu_k + \bar{O}^2_k(\gamma)\mu_k + \bar{O}^3_k(\gamma)\big)}{\partial \mu_k}=0$.
	
	Using \eqref{OptimalSolution}, the optimal solution of \eqref{OptimizationProblem8_2} is obtained as
	\begin{align}
		\label{OptimalSolution1}
		M(\gamma)=\bar{O}^3_k(\gamma)-\frac{1}{4}\bar{O}^2_k(\gamma)\big(\bar{O}^1_k(\gamma)\big)^{-1}\big(\bar{O}^2_k(\gamma)\big)^T.
	\end{align}	
	
	This implies that for a given value $\gamma\in [\underline{\gamma},\overline{\gamma}]$, $M(\gamma)$ can be efficiently computed by \eqref{OptimalSolution1}. Thus, at time instant $k$, a bisection method is summarized in \algref{OptimalFD} to search $\gamma^*_k$ by solving a series of problems \eqref{OptimizationProblem8_2} for different values of $\gamma \in [\underline{\gamma},\overline{\gamma}]$. Once $\gamma^*_k$ is searched, $\mu^*_k$ can be computed by \eqref{OptimalSolution}.
	Finally, optimal gains $L^{i,*}_k$ ($\forall i \in \mathbb{I}_u$) and inputs $u^*_k$ can be obtained based on $\mu^*_k$ and \eqref{AugmentedVariable}.
	\begin{algorithm}[htbp]
		\caption{Computation of $L^{i,*}_k$ and $u^*_k$}
		\label{OptimalFD}
		\begin{algorithmic}[1]
			\STATE At time instant $k$, given $\underline{\gamma}=0$ and a precision $\epsilon$;
			\STATE Compute $\overline{\gamma}$ by solving \eqref{OptimizationProblem8_Gamma_1};
			\STATE Compute $\gamma=\frac{\underline{\gamma}+\overline{\gamma}}{2}$ and solve \eqref{OptimalSolution1};
			\WHILE{$\vert M(\gamma)\vert> \epsilon$}
			\IF{$M(\gamma)>0$}
			\STATE Set $\underline{\gamma}=\gamma$;
			\ELSE
			\STATE Set $\overline{\gamma}=\gamma$;
			\ENDIF
			\STATE Repeat Step 3 to update $\gamma$ and $M(\gamma)$;
			\ENDWHILE
			\STATE Set $\gamma^*_k=\gamma$;
			\STATE Solve $M(\gamma^*_k)$ in \eqref{OptimizationProblem8_2} by \eqref{OptimalSolution} to obtain $\mu^{*}_k$;
			\STATE Compute $L^{i,*}_k$ ($\forall i \in \mathbb{I}_u$) and $u^*_k$ by \eqref{AugmentedVariable};
			\RETURN $\gamma^*_k$, $\mu^{*}_k$,  $L^{i,*}_k$ ($\forall i \in \mathbb{I}_u$) and $u^*_k$.
		\end{algorithmic}
	\end{algorithm}
	
	\begin{remark}
		In the unconstrained case, the optimization problem has its analytical solution and can be solved efficiently. But designed gains cannot guarantee the SVO stability and we have to further consider the constrained case to design gains with SVO stability guarantees.
	\end{remark}

	\subsection{Solution of Constrained Problem}
	\label{Section4_5}
	In the preliminary results \cite{Xu2022CDC}, a matrix inequality stability condition was established for the SVO. Based on it, optimal gains stabilizing the SVO were computed by a branch and bound method. However, it has quite high computational complexity and is not appropriate for online applications. In order to guarantee the SVO stability and simultaneously reduce computational complexity, we consider solving the constrained case of \eqref{FD6} with a new stability condition established for joint design of gains and inputs here. On one hand, we consider that inputs are energy-bounded, i.e.,
	$u_k \in \mathcal{U}=\{u_k:\Vert u_k-u^c\Vert_2 \le \bar{u}\}$. On the other hand, a new SVO stability condition for gains is established below.
	\begin{proposition}
		\label{FD9}
		The $i$-th SVO \eqref{Observer1} is stable if
		\begin{align}
			\label{OptimizationProblem16_2}
			\mathrm{vec}(L^i_k) \in \mathcal{L}_i=&\{ \mathrm{vec}(L^i_k):\Vert\mathrm{vec}(A)-(C^T\otimes I)\mathrm{vec}(L^i_k)\Vert^2_2 \notag\\&\le 1-\epsilon_3\}, \: \forall k \ge 0,
		\end{align}	
		where $\epsilon_3$ is a sufficiently small positive scalar.
	\end{proposition}
	\begin{proof}
		For a dynamics  $\mathring{x}_{k+1}=(A-L^{i}_kC)\mathring{x}_k + \delta_k$, when $\delta_k$ is bounded, the dynamics is stable if $\mathring{x}_{k+1}=(A-L^{i}_kC)\mathring{x}_k$ is stable. Furthermore, if $\Vert A-L^i_kC \Vert_2 \le 1$,
		\begin{align}
			\label{OptimizationProblem16_1}
			\Vert\mathring{x}_{k+1}\Vert_2 =& \Vert(A-L^{i}_kC)\mathring{x}_k\Vert_2 \notag\\\le& \Vert(A-L^{i}_kC)\Vert_2\Vert\mathring{x}_k\Vert_2 \le \Vert\mathring{x}_k\Vert_2.
		\end{align}
		
		If \eqref{OptimizationProblem16_1} holds $\forall k \ge 0$, it means that the dynamics is stable. Besides, since $\Vert A-L^i_kC \Vert_2 \le \Vert A-L^i_kC \Vert_F$, if $\Vert A-L^i_kC \Vert^2_F \le 1-\epsilon_3$, the dynamics is stable as well. Moreover, $\Vert A-L^i_kC \Vert^2_F=\Vert \mathrm{vec}(A-L^i_kC)\Vert^2_2=\Vert\mathrm{vec}(A)-(C^T\otimes I)\mathrm{vec}(L^i_k)\Vert^2_2$. This means that the dynamics is stable if \eqref{OptimizationProblem16_2} holds, which implies that the SVO \eqref{Observer1} is stable.
	\end{proof}
	
	Under \propref{FD9}, considering a general case that the SVO should be stable and that inputs are bounded, it implies the following constraint for \eqref{FD6}:
	\begin{align}
		\label{AugmentedSet}
		\mu_k \in \mathcal{S}= \underbrace{\mathcal{L}_0\times \mathcal{L}_1\times\cdots\times\mathcal{L}_{n_u}}^{n_u+1} \times\mathcal{U} .
	\end{align}	
	
	\begin{remark}
		For brevity, the same notation $\mathcal{S}$ is used in \eqref{AugmentedSet} with that in Section \ref{Section4_4}. Differently, $\mathcal{S}$ in \eqref{AugmentedSet} represents a constrained set for the SVO stability condition and the boundedness of inputs. For this reason, the same notation $\mathcal{S}$ is used here, the optimization problems in this subsection have the same expressions with some optimization problems in Section \ref{Section4_4}. However, without ambiguity, the optimization problems here correspond to the constrained case \eqref{AugmentedSet}.
	\end{remark}
	
	We will further propose a method to handle \eqref{OptimizationProblem8_2} with \eqref{AugmentedSet} directly. It is observed that under the constrained case, a key point of the bisection method in \algref{OptimalFD} is to solve $M(\gamma)$ in \eqref{OptimizationProblem8_2} with \eqref{AugmentedSet} for a given $\gamma$. In \algref{OptimalFD}, when the constrained case is considered, $M(\gamma)$ in \eqref{OptimizationProblem8_2} with \eqref{AugmentedSet} is possible to be an indefinite quadratic programming (IQP) problem for a given $\gamma$ (i.e., $\bar{O}^1_k(\gamma)$ is indefinite). Thus, we consider a more general case that $\bar{O}^1_k(\gamma)$ is indefinite here. In this case, we propose to compute an approximate solution of \eqref{OptimizationProblem8_2} with \eqref{AugmentedSet}, which constructs a convex function to approximate the indefinite quadratic objective function with a given precision such that the original problem can be approximately solved under the convex framework \cite{Yamamoto2007}.
	
	It is observed that $\bar{O}^1_k(\gamma)$ is symmetric. Therefore, based on the congruent transformation, we could find an orthogonal matrix $D$ such that $D^T\bar{O}^1_k(\gamma)D$ is a diagonal matrix. This implies that using the coordinate transformation $\mu_k=Dv_k$, the objective function of \eqref{OptimizationProblem8_2} is transformed into $\Phi(\gamma,v_k)=v^T_k\Theta^1_kv_k + \Theta^2_kv_k + \Theta^3_k$,
	where $\Theta^1_k=D^T\bar{O}^1_k(\gamma)D$, $\Theta^2_k=\bar{O}^2_k(\gamma)D$ and $\Theta^3_k=\bar{O}^3_k(\gamma)$. Thus, the optimization problem \eqref{OptimizationProblem8_2} is transformed into
	\begin{align}
		\label{OptimizationProblem8_3}
		\Phi(\gamma)=\min_{v \in \mathcal{S}_{v}}\:v^T\Theta^1_kv + \Theta^2_kv + \Theta^3_k,
	\end{align}
	where $\mathcal{S}_{v}$	denotes the constraint set of $v$ after the coordinate transformation $\mu_k=Dv_k$ from $\mathcal{S}$. Since $\Theta^1_k$ is a diagonal matrix, we could transform $\Phi(\gamma,v_k)$ into
	\begin{align}
		\label{PiecewiseConstraint0_0}
		\Phi(\gamma,v_k)=f_1(\gamma,v_k) + f_2(\gamma,v_k)
	\end{align}
	with $f_1(\gamma,v_k)=\sum^{t}_{l=1}\theta^{1,l}_k(v^l_k)^2$ and $f_2(\gamma,v_k)=\sum^{n_xn_y(n_u+1)+n_u}_{l=t+1}\theta^{1,l}_k(v^l_k)^2 + \Theta^2_kv_k + \Theta^3_k$.
	where $\theta^{1,l}_k$ is the $l$-th diagonal element of $\Theta^1_k$, $v^l_k$ is the $l$-th element of $v_k$ and $t$ denotes the number of negative diagonal elements of $\Theta^1_k$ (i.e., $\theta^{1,l}_k< 0$ for $l=1,2,\cdots,t$ and $\theta^{1,l}_k> 0$ for $l=t+1,t+2,\cdots,n_xn_y(n_u+1)+n_u$).
	
	Note that only $\theta^{1,l}_k\not=0$ for $l=1,2,\cdots,n_xn_y(n_u+1)+n_u$ are considered. $f_1(\gamma,v_k)$ and $f_2(\gamma,v_k)$ are the concave and convex components of $\Phi(\gamma,v_k)$, respectively. To solve \eqref{OptimizationProblem8_3}, the concave part $f_1(\gamma,v_k)$ is represented as
	\begin{align}
		\label{PiecewiseConstraint0_1}
		f_1(\gamma,v_k)=\sum^{t}_{l=1}f^l_1(\gamma,v_k),
	\end{align}
	where $f^l_1(\gamma,v_k)=\theta^{1,l}_k(v^l_k)^2$ with $l=1,2,\cdots,t$.
	
	In the following, $f_1(\gamma,v_k)$ can be approximated by a piecewise linear function such that the concavity of $f(\gamma,v_k)$ is finally eliminated. Since $v_k=D^{-1}\mu_k$, the intervals of the components of $v_k$ are obtained by using the interval of $\mu_k$. The interval of $v^l_k$ is denoted as $[\underline{v}^l, \overline{v}^l]$. In order to approximate $f^l_1(\gamma,v_k)$, the interval $[\underline{v}^l, \overline{v}^l]$ of $v^l_k$ is divided into $m$ small sub-intervals with the same width, i.e.,
	\begin{align}
		\label{IntervalPartition}
		[\underline{v}^l, \overline{v}^l]=\bigcup_{s=1}^{m}[v^{ls}, v^{l(s+1)}],
	\end{align}	
	where $[v^{ls}, v^{l(s+1)}]$ is the $s$-th small sub-interval and $v^{ls}$ ($s=1,2,\cdots,m+1$) are the $m+1$ ending points of the $m$ small sub-intervals with $v^{l1}=\underline{v}^l$ and $v^{l(m+1)}=\overline{v}^l$.
	
	Based on \eqref{IntervalPartition} and \cite{Yamamoto2007}, $f^l_1(\gamma,v)$ is approximated by
	\begin{align}
		\label{PiecewiseConstraint0}
		\tilde{f}^l_1(\gamma,\zeta)= \theta^{1l}_k(\sum^{m+1}_{s=1}\zeta^{ls}(v^{ls})^2)
	\end{align}
	with
	\begin{subequations}
		\begin{align}
			\label{PiecewiseConstraint}
			&0\le \zeta^{ls} \le 1, \sum^{m+1}_{s=1}\zeta^{ls} =1,\zeta^{l1} \le  \nu^{l1}, \\
			&\zeta^{ls} \le  \nu^{ls} + \nu^{l(s-1)}\:\:\text{for}\:\:s=2,3,\cdots,m, \\
			&\zeta^{l(m+1)} \le  \nu^{lm}, \nu^{ls}=0 \:\text{or}\:1, \sum^{m}_{s=1}\nu^{ls} =1,
		\end{align}
	\end{subequations}
	where $\tilde{f}^l_1(\gamma,\zeta)$ provides a piecewise lower bound for $f^l_1(\gamma,v_k)$ for $v^l_k \in [\underline{v}^l, \overline{v}^l]$. In each small sub-interval of $[\underline{v}^l, \overline{v}^l]$, a linear function is used to approximate $f^l_1(\gamma,v_k)$. The approximate precision is determined by the value of $m$. The larger $m$ is, the higher the precision is. However, computational complexity becomes higher. As $m\rightarrow \infty$, \eqref{PiecewiseConstraint0} becomes equivalent to \eqref{PiecewiseConstraint0_1}. Based on the analysis above, for all the components $f^l_1(\gamma,v_k)$ with $v^l_k \in [\underline{v}^l, \overline{v}^l]$ ($l=1,2,\cdots,t$), we could use the same method to make a piecewise linear approximation. Consequently, for $v_k \in [\underline{v}, \overline{v}]$, $f(\gamma,v_k)$ can be approximated by a linear function:
	\begin{align}
		\label{ApproximateFunction}
		\tilde{f}(\gamma,v_k,\zeta)= \sum^{t}_{l=1}\sum^{m+1}_{s=1}\theta^{1l}_k\zeta^{ls}(v^{ls})^2 + f_2(\gamma,v_k).
	\end{align}
	
	Based on \eqref{PiecewiseConstraint0_1}-\eqref{ApproximateFunction}, under the transformation $\mu_k=Dv_k$, \eqref{OptimizationProblem8_3} is approximated by the optimization problem:
	\begin{align}
		\label{LowerBound}
		&\min_{v_k \in \mathcal{S}_{v},\zeta^{ls},\nu^{ls}} \tilde{f}(\gamma,v_k,\zeta),\: s.t.\:\eqref{PiecewiseConstraint},\:l=1,2,\cdots,t,
	\end{align}	
	whose solution is denoted as $v^*_k$, $\zeta^{ls,*}$ and $\nu^{ls,*}$. An approximate solution of \eqref{OptimizationProblem8_2} with \eqref{AugmentedSet} is obtained by $$\mu_k^*=Dv^*_k.$$ 
	
	Notice that if we want to follow the bisection method in Section \ref{Section4_4} to solve \eqref{OptimalGain3} with the constraint \eqref{AugmentedSet}, the key point is still to find a bound of $\gamma$ (i.e., the upper bound of $\gamma$). Different from the unconstrained case, $\mathcal{S}$ in \eqref{AugmentedSet} is not a sufficiently large set, which means that we cannot directly use \eqref{OptimizationProblem8_Gamma_1} to find an upper bound for $\gamma$ in the constrained case. Thus, the method proposed in \cite{Xu2022CDC} is still used to search an upper bound for $\gamma$ here. Particularly, at time instant $k$, we propose to use the known $\gamma^*_{k-1}$ obtained by the bisection method at time instant $k-1$ and define a step
	length at time instant $k$ as
	\begin{align}
		\label{StepLength}
		\mathring{\gamma}_k=\gamma^*_{k-1},\:k\ge 1,
	\end{align}	
	where $\mathring{\gamma}_0$ is defined as $\mathring{\gamma}_0=\arg\max \:\gamma,\:\mathrm{s.t.}\:\bar{O}^1_k(\gamma)\succeq 0$ (see \eqref{OptimizationProblem8_Gamma_1}).
	Since $\gamma^*_k \ge 0$, $[0,\mathring{\gamma}_k]$ can be used as an initial interval at time instant $k$ for searching an interval
	$[\underline{\gamma},\overline{\gamma}]$ with a width $\mathring{\gamma}_k$ for $\gamma^*_k$ , which is summarized in \algref{OptimalFD1}. Therefore, the problem \eqref{OptimalGain3} with \eqref{AugmentedSet} can be approximately solved by \algref{OptimalFD} as well, where the procedures for searching a bound $[\underline{\gamma},\overline{\gamma}]$ of $\gamma^*_k$ and for solving \eqref{OptimizationProblem8_2} with \eqref{AugmentedSet} are replaced by \algref{OptimalFD1} and the problem \eqref{LowerBound}, respectively.
	\begin{algorithm}[htbp]
		\caption{Searching an interval $[\underline{\gamma},\overline{\gamma}]$ for $\gamma^*_k$}
		\label{OptimalFD1}
		\begin{algorithmic}[1]
			\STATE At time instant $k$, obtain $\mathring{\gamma}_k$ in \eqref{StepLength} or $\mathring{\gamma}_0$ if $k=0$;
			\STATE Given $\underline{\gamma}= 0$ and $\overline{\gamma}=\mathring{\gamma}_k$;
			\WHILE{$M(\overline{\gamma})>0$}
			\STATE Set $\underline{\gamma}=\underline{\gamma}+\mathring{\gamma}_k$;
			\STATE Set $\overline{\gamma}=\overline{\gamma}+\mathring{\gamma}_k$;
			\ENDWHILE
			\RETURN $\underline{\gamma}$ and $\overline{\gamma}$.
		\end{algorithmic}
	\end{algorithm} 	
	
	\begin{remark}
		For a given $\gamma$, only when \eqref{OptimizationProblem8_2} with \eqref{AugmentedSet} is an IQP problem, the proposed method here is used to solve \eqref{OptimizationProblem8_2}. Although \eqref{LowerBound} provides an approximate solution, as long as $m$ is set as a sufficiently large value, \eqref{LowerBound} can provide an arbitrarily accurate solution for \eqref{OptimizationProblem8_2} and \eqref{OptimizationProblem8_3} with \eqref{AugmentedSet}. Moreover, we only need to use \eqref{LowerBound} and $\mu_k^*=Dv^*_k$ to replace \eqref{OptimizationProblem8_2} in \algref{OptimalFD} and a joint gain and input design in the constrained case can be achieved. Thus, $m$ is a key parameter to adjust computational complexity and FD conservatism of the proposed method here. 
	\end{remark}

	\section{Illustrative Examples}	
	An example with two actuators in the form \eqref{PlantDynamics1} is used to illustrate the effectiveness of the proposed methods. The system parametric matrices are given as
	\begin{align*}
		A=&\begin{bmatrix} 0.5 & 0.3 \\0.2 & 0.6 \end{bmatrix}, B=\begin{bmatrix} 0.05 & 0.08 \\0.07 & 0.05 \end{bmatrix}, E=\begin{bmatrix} 0.05 & 0.03 \\0.04 & 0.05 \end{bmatrix}, \\
		C=&\begin{bmatrix} 1 & 0 \\0 & 1\end{bmatrix},
		F=\begin{bmatrix} 0.1 & 0 \\0 & 0.1 \end{bmatrix}.
	\end{align*}
	
	Consider the healthy and faulty modes $G_0=I$ and 
	\begin{align*}
		\mathbf{G}_1=\begin{bmatrix} [0,0.8] & 0 \\0 & 1 \end{bmatrix}, \mathbf{G}_2=\begin{bmatrix} 1 & 0 \\0 & [0, 0.8] \end{bmatrix},
	\end{align*}
	where both $\mathbf{G}_1$ and $\mathbf{G}_2$ denote a fault interval $[0, 0.8]$ in the two actuators, respectively. The initial state, disturbances and noises are bounded by the zonotopes:
	\begin{align*}
		\hat{\mathcal{X}}_0=&\hat{\mathcal{X}}_1=\hat{\mathcal{X}}_2=\langle \begin{bmatrix} 0.55\\0.55 \end{bmatrix},\begin{bmatrix} 0.5 & 0 \\0 & 0.5 \end{bmatrix}\rangle,\\
		\mathcal{W}=&\langle \begin{bmatrix} 0\\0 \end{bmatrix},\begin{bmatrix} 0.5 & 0 \\0 & 0.5 \end{bmatrix}\rangle, \mathcal{V}=\langle \begin{bmatrix} 0\\0 \end{bmatrix},\begin{bmatrix} 0.1 & 0 \\0 & 0.1 \end{bmatrix}\rangle.
	\end{align*}

	\subsection{Gain Design for PFD}
	\label{5_1}
	We consider three scenarios to illustrate the effectiveness of the proposed gain optimization method for PFD: 
	\begin{itemize}
		\item The first scenario does not consider the stability condition in \propref{FD9};
		\item The second scenario considers the stability condition in \propref{FD9};
		\item The third scenario compares the proposed method in this paper with the method proposed in the preliminary results \cite{Xu2022CDC}. For fairness, both of them use the same stability condition in \propref{FD9}.
	\end{itemize}
	
	In the first scenario, three SVOs for the modes $G_0$, $\mathbf{G}_1$ and $\mathbf{G}_2$ are designed, where each SVO matches a mode. The specific faults in the two actuators are set as
	\begin{align*}
		G_1=\begin{bmatrix} 0.55 & 0 \\0 & 1 \end{bmatrix}, G_2=\begin{bmatrix} 1 & 0 \\0 & 0.55 \end{bmatrix}.
	\end{align*}  
	
	We have done two simulations to test the effectiveness of the proposed method. As given above, all the three SVOs use the same initial state set and thus the sets at $k=0$ are omitted in \figsref{uc_G1_PFD_NoStability} and \ref{uc_G2_PFD_NoStability} to save space. The initial system state is set as $x_0=\begin{bmatrix} 0.6 & 0.6 \end{bmatrix}^T$. In the first simulation, constant inputs $u=\begin{bmatrix} -0.7 & 3 \end{bmatrix}^T$ are injected into the system. We test the diagnosis of the first actuator fault $G_1$, which is shown in \figref{uc_G1_PFD_NoStability}. It is shown that the first actuator fault $G_1$ is diagnosed at $k=1$ because $\mathbf{0} \in \mathcal{R}^1_1$, $\mathbf{0} \not\in \mathcal{R}^0_1$ and $\mathbf{0} \not\in \mathcal{R}^2_1$ are tested. Note that the notations $R(l)$ ($l=1,2$) in \figref{uc_G1_PFD_NoStability} denote the first and second components of residual zonotopes, respectively. Similarly, in the second simulation, constant inputs $u=\begin{bmatrix} 3 & -0.7 \end{bmatrix}^T$ are injected into the system. We test the diagnosis of the second actuator fault $G_2$, which is shown in \figref{uc_G2_PFD_NoStability}. Differently, it is shown that the second actuator fault $G_2$ is successfully diagnosed at $k=13$ because $\mathbf{0} \not\in \mathcal{R}^0_{13}$ and $\mathbf{0} \in \mathcal{R}^2_{13}$ are tested.
	\begin{figure}[htbp]
		\centering
		\subfigure[$k=1$]{\includegraphics[height=4.5cm,width=8.5cm]{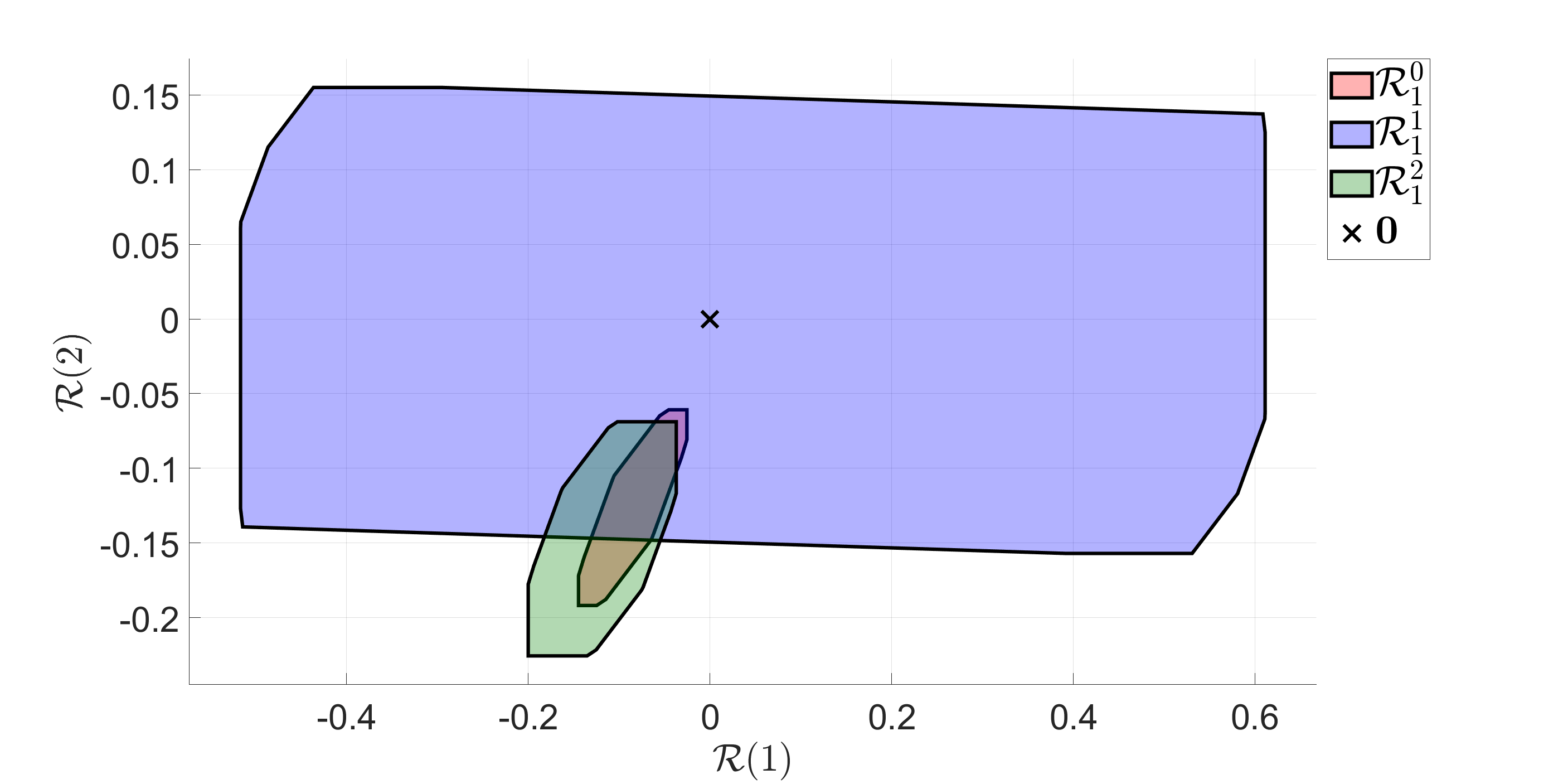}}
		\caption{PFD of the fault $G_1$ without considering stability}
		\label{uc_G1_PFD_NoStability}
	\end{figure}
	\begin{figure}[htbp]
		\centering
		\subfigure[$k=1$]{\includegraphics[height=4.5cm,width=8.5cm]{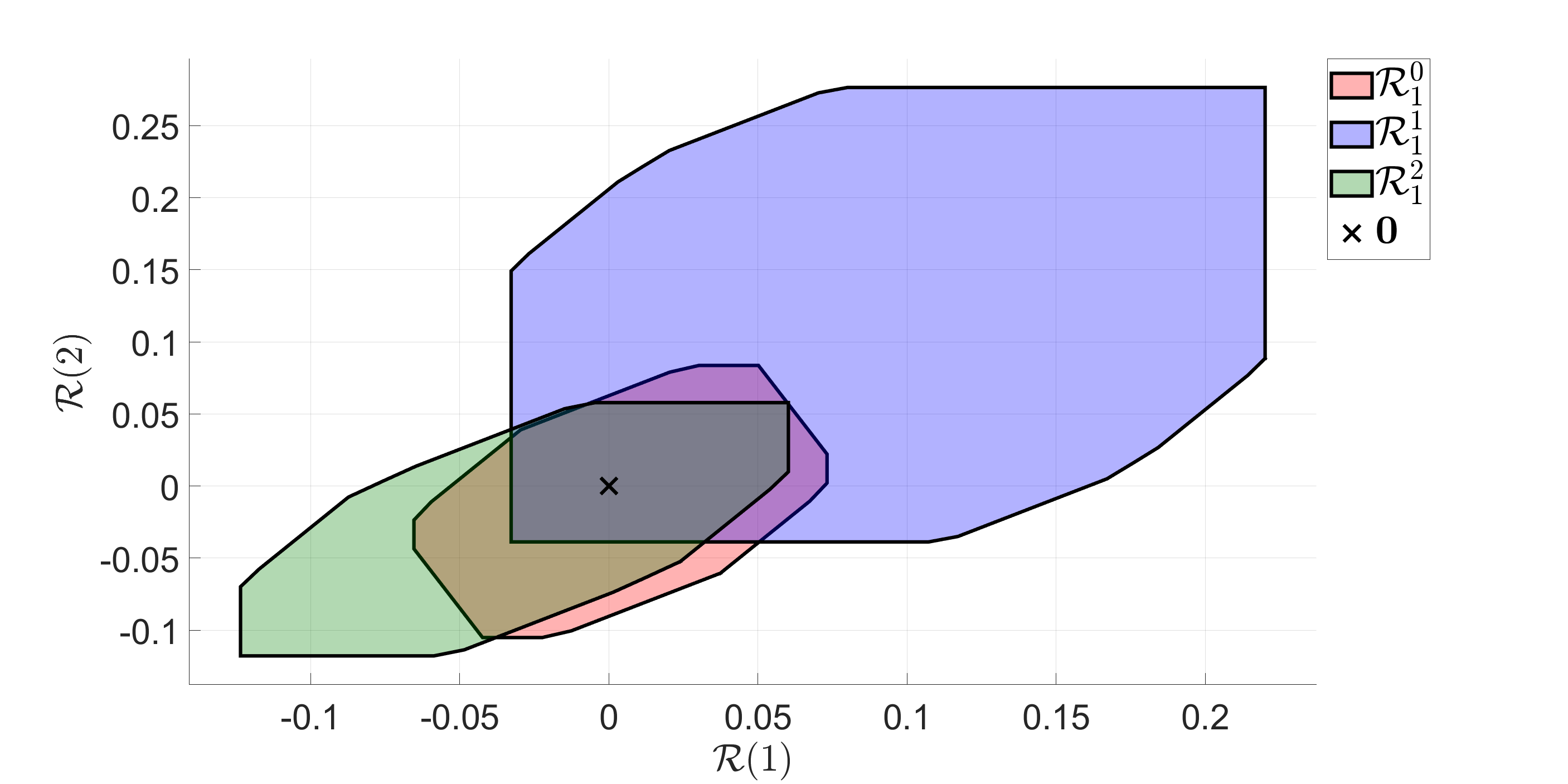}}
		\subfigure[$k=10$]{\includegraphics[height=4.5cm,width=8.5cm]{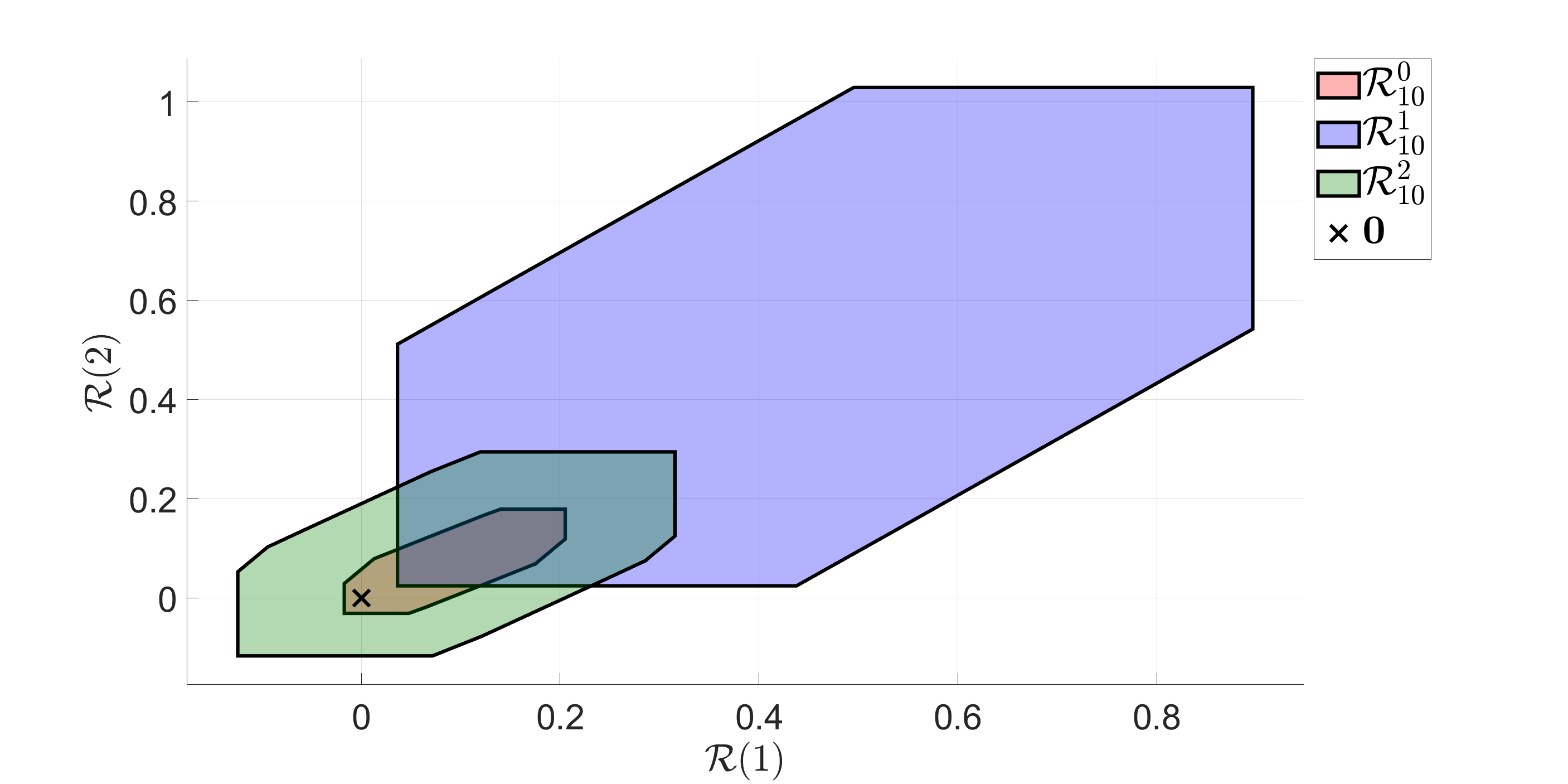}}
		\subfigure[$k=13$]{\includegraphics[height=4.5cm,width=8.5cm]{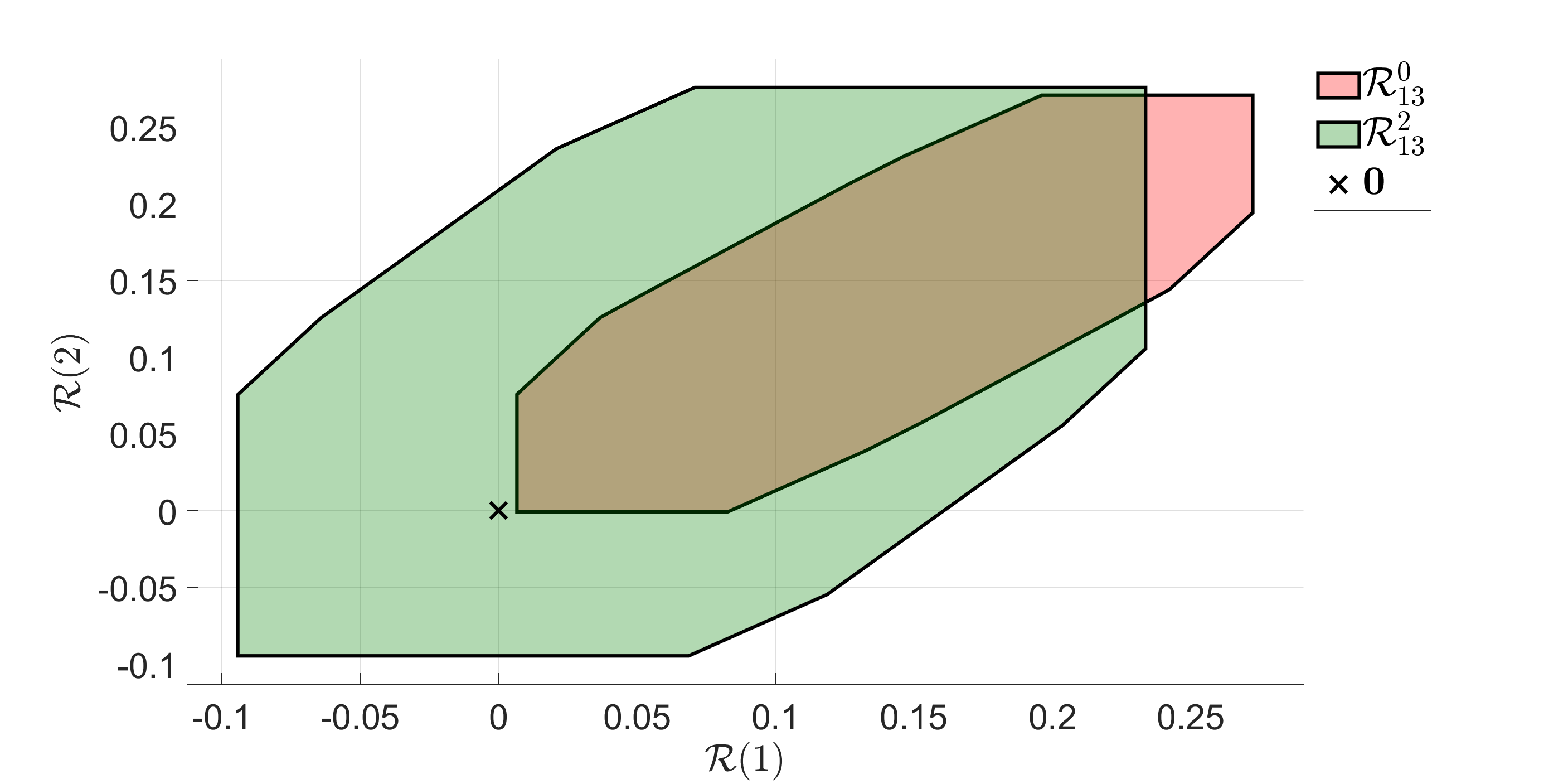}}
		\caption{PFD of the fault $G_2$ without considering stability}
		\label{uc_G2_PFD_NoStability}
	\end{figure}

	In the second scenario, the stability condition given in \propref{FD9} is considered as a constraint when designing gains for the three SVOs. The parameter $\epsilon_3$ in \propref{FD9} is set as $0.001$. Therefore, the stability of all the three SVOs can be guaranteed. However, due to the restriction of the stability constraint, the designed optimal gains have to lose some FD performance in this case. Similar to the first scenario, the second scenario also considers two simulations for diagnosis of the faults in the first and second actuators, respectively. All the parameters, inputs and faults are the same with those given in the first scenario and the only difference consists in the added stability constraint. The simulation results of diagnosis of the faults $G_1$ and $G_2$ are shown in \figsref{c_G1_PFD_Stability} and \ref{c_G2_PFD_Stability}. In \figref{c_G1_PFD_Stability}, it is tested that $\mathbf{0} \not\in \mathcal{R}^0_{13}$ and $\mathbf{0} \in \mathcal{R}^1_{13}$, which imply the diagnosis of the fault $G_1$ in the first actuator at $k=13$. Compared with the first scenario, more diagnosis time is needed. Similarly, in \figref{c_G2_PFD_Stability}, it is shown that the fault $G_2$ in the second actuator is diagnosed at $k=15$ (i.e., $\mathbf{0} \not\in \mathcal{R}^0_{15}$ and $\mathbf{0} \in \mathcal{R}^2_{15}$). These results have assessed the effectiveness of the proposed gain optimization method for PFD.
	\begin{figure}[htbp]
		\centering
		\subfigure[$k=1$]{\includegraphics[height=4.5cm,width=8.5cm]{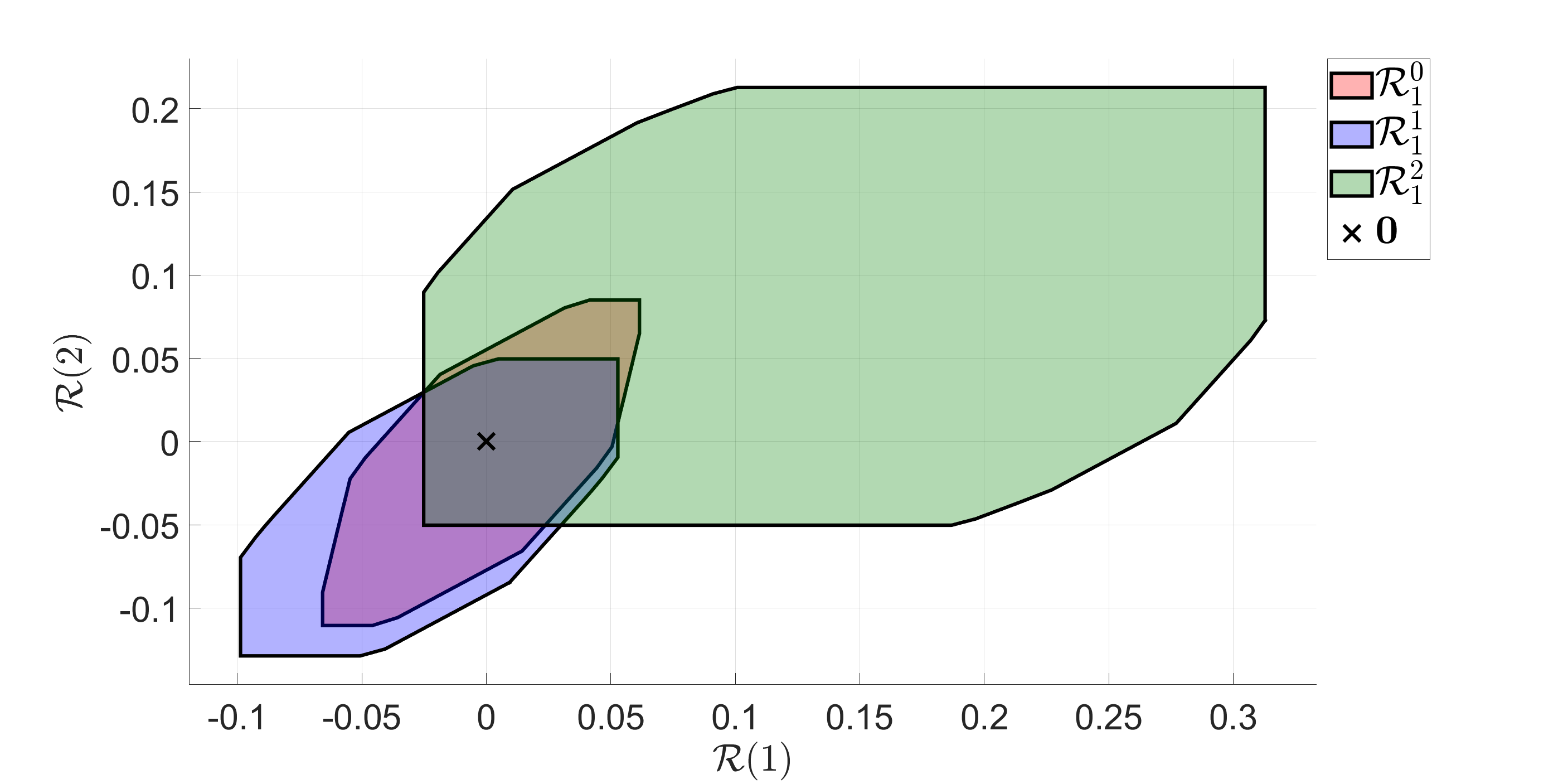}}
		\subfigure[$k=10$]{\includegraphics[height=4.5cm,width=8.5cm]{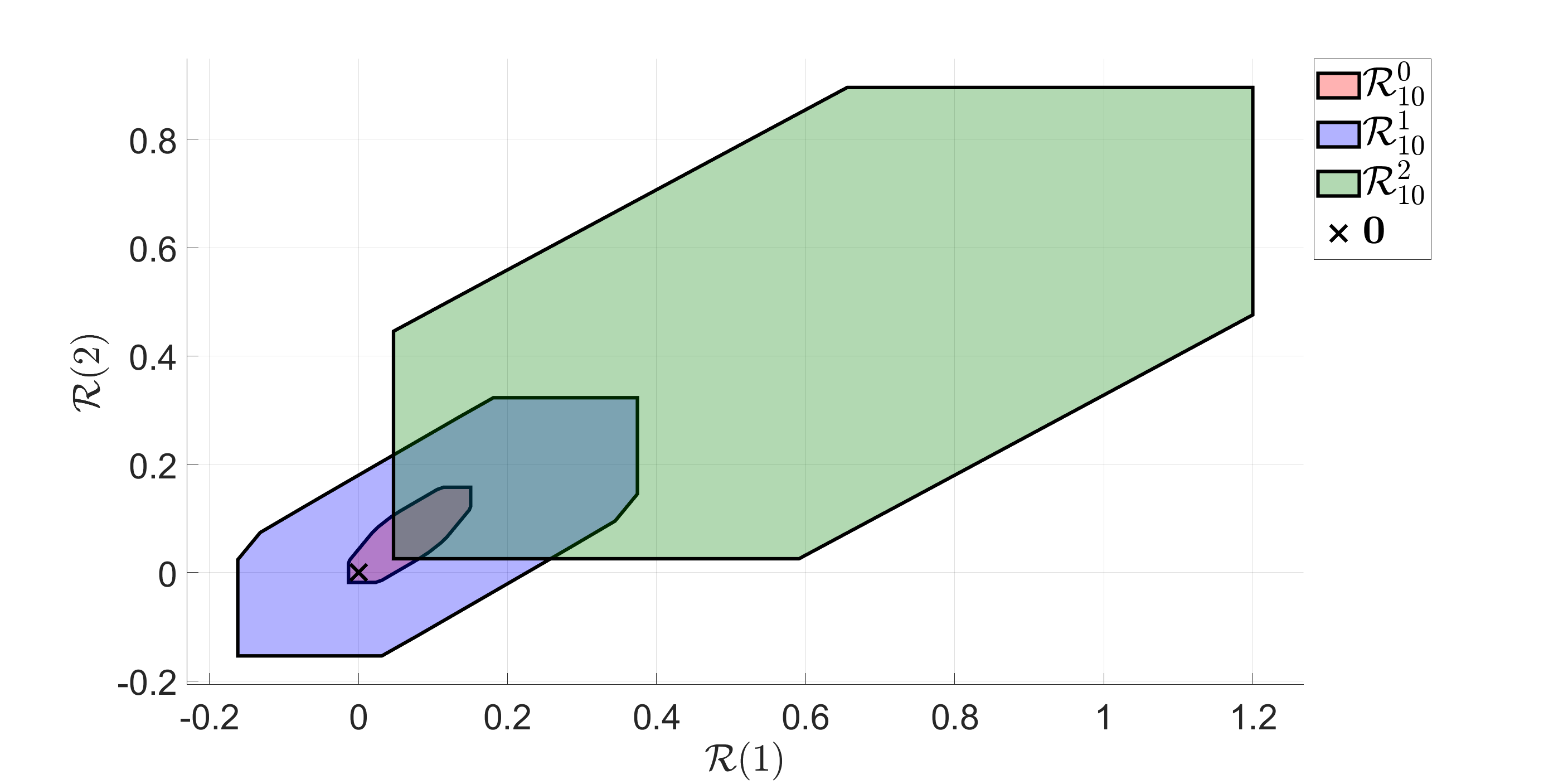}}
		\subfigure[$k=13$]{\includegraphics[height=4.5cm,width=8.5cm]{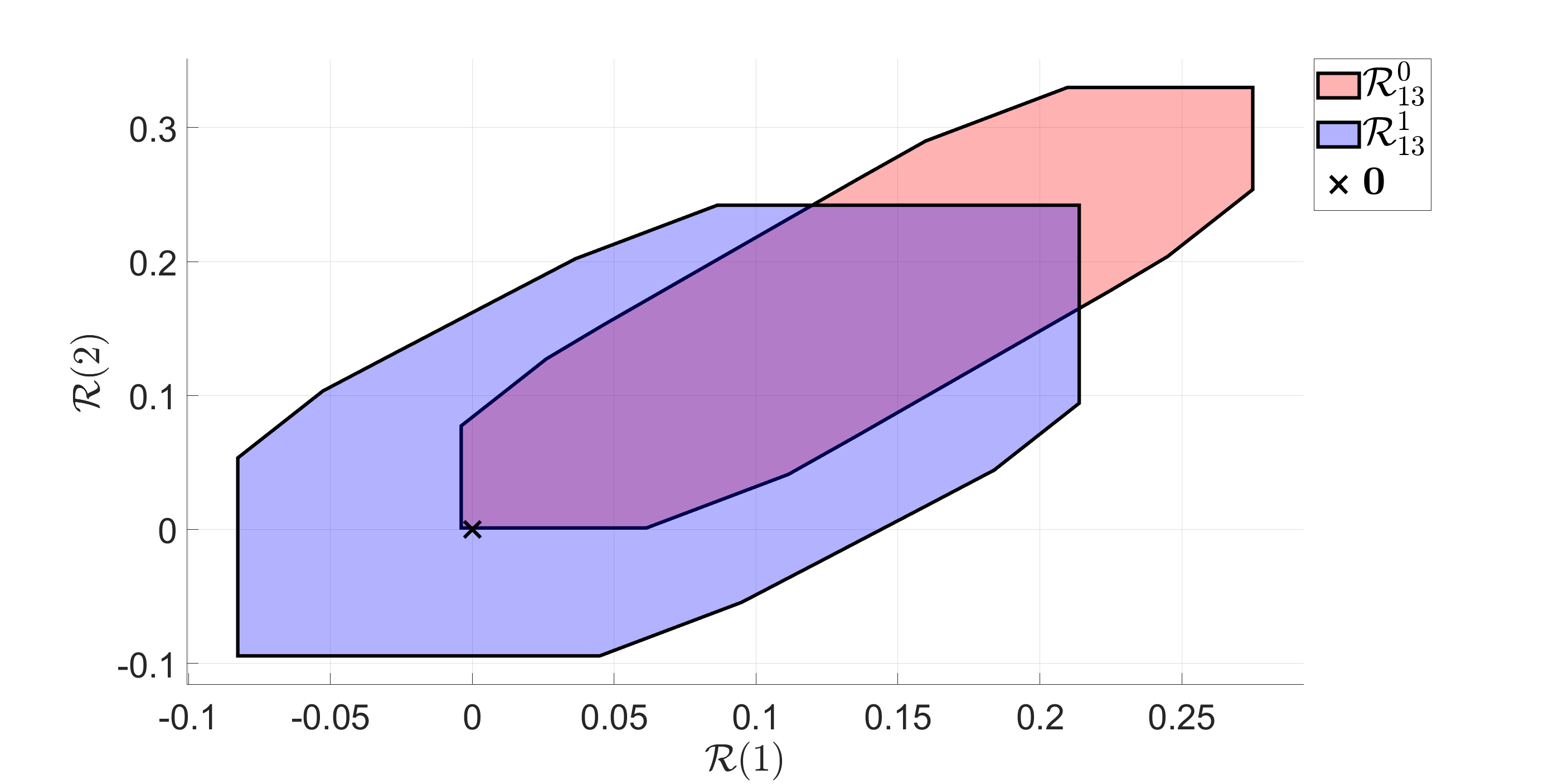}}
		\caption{PFD of the first fault $G_1$ with considering stability}
		\label{c_G1_PFD_Stability}
	\end{figure}
	\begin{figure}[htbp]
		\centering
		\subfigure[$k=1$]{\includegraphics[height=4.5cm,width=8.5cm]{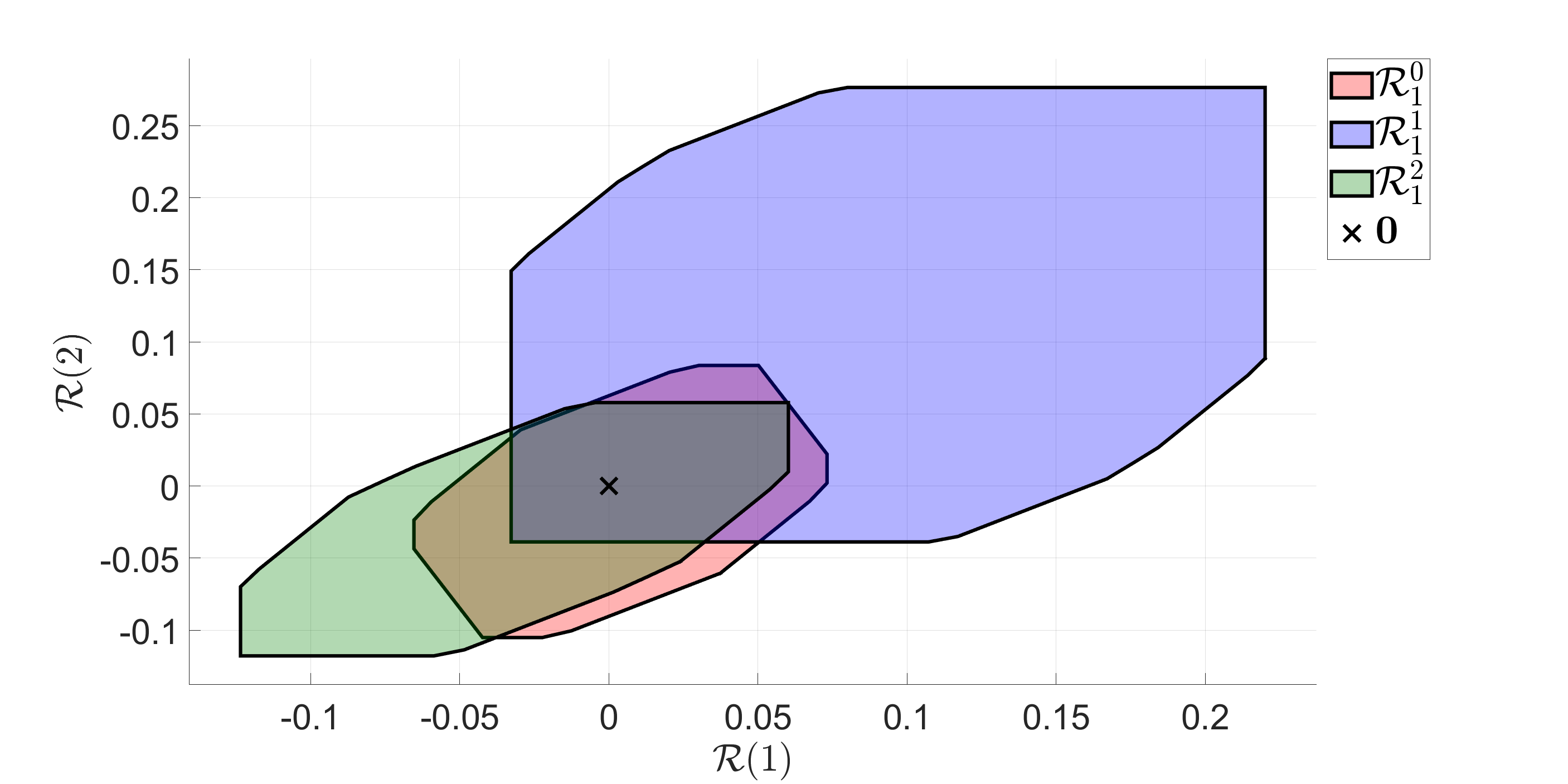}}
		\subfigure[$k=2$]{\includegraphics[height=4.5cm,width=8.5cm]{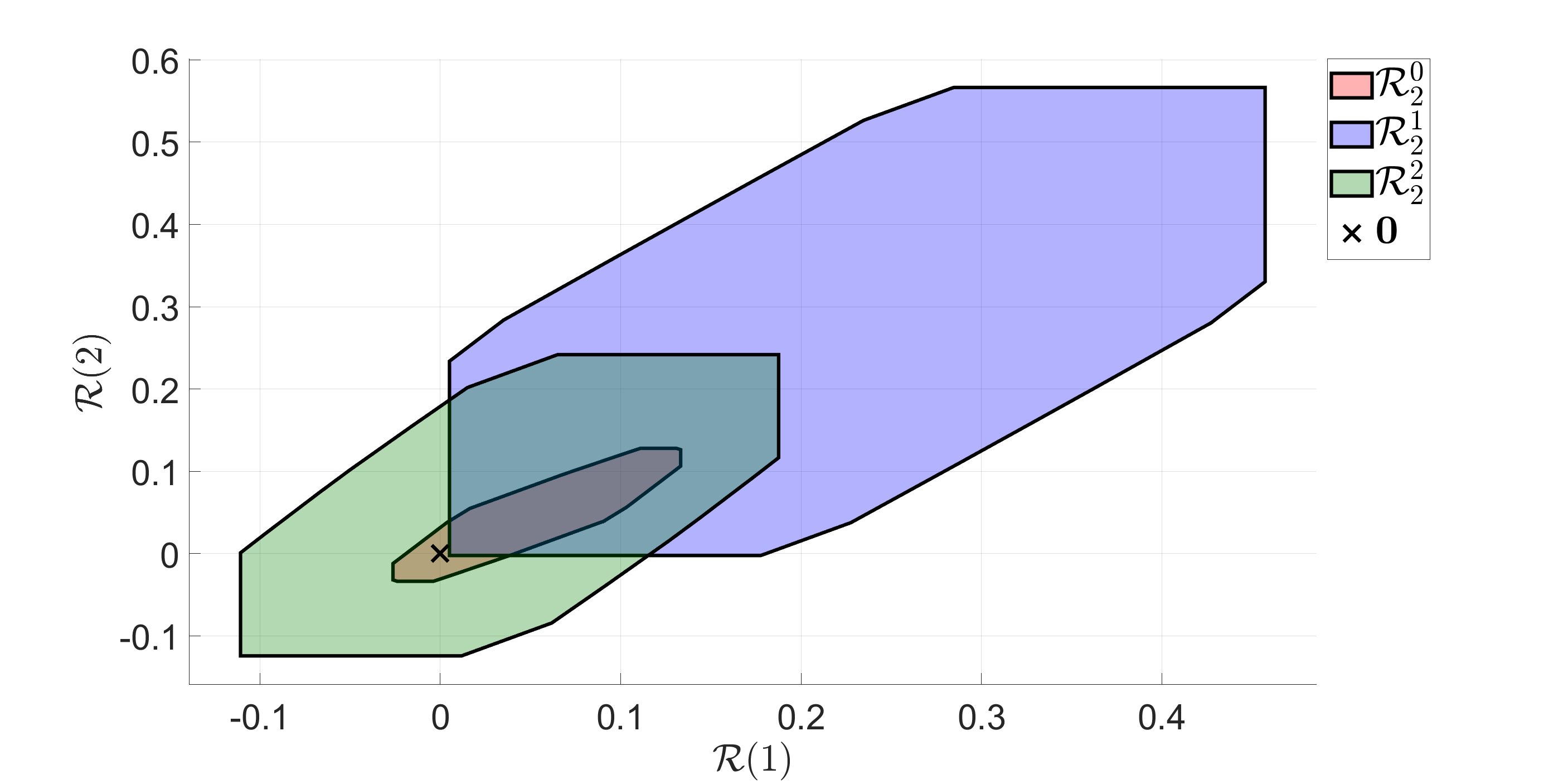}}
		\subfigure[$k=15$]{\includegraphics[height=4.5cm,width=8.5cm]{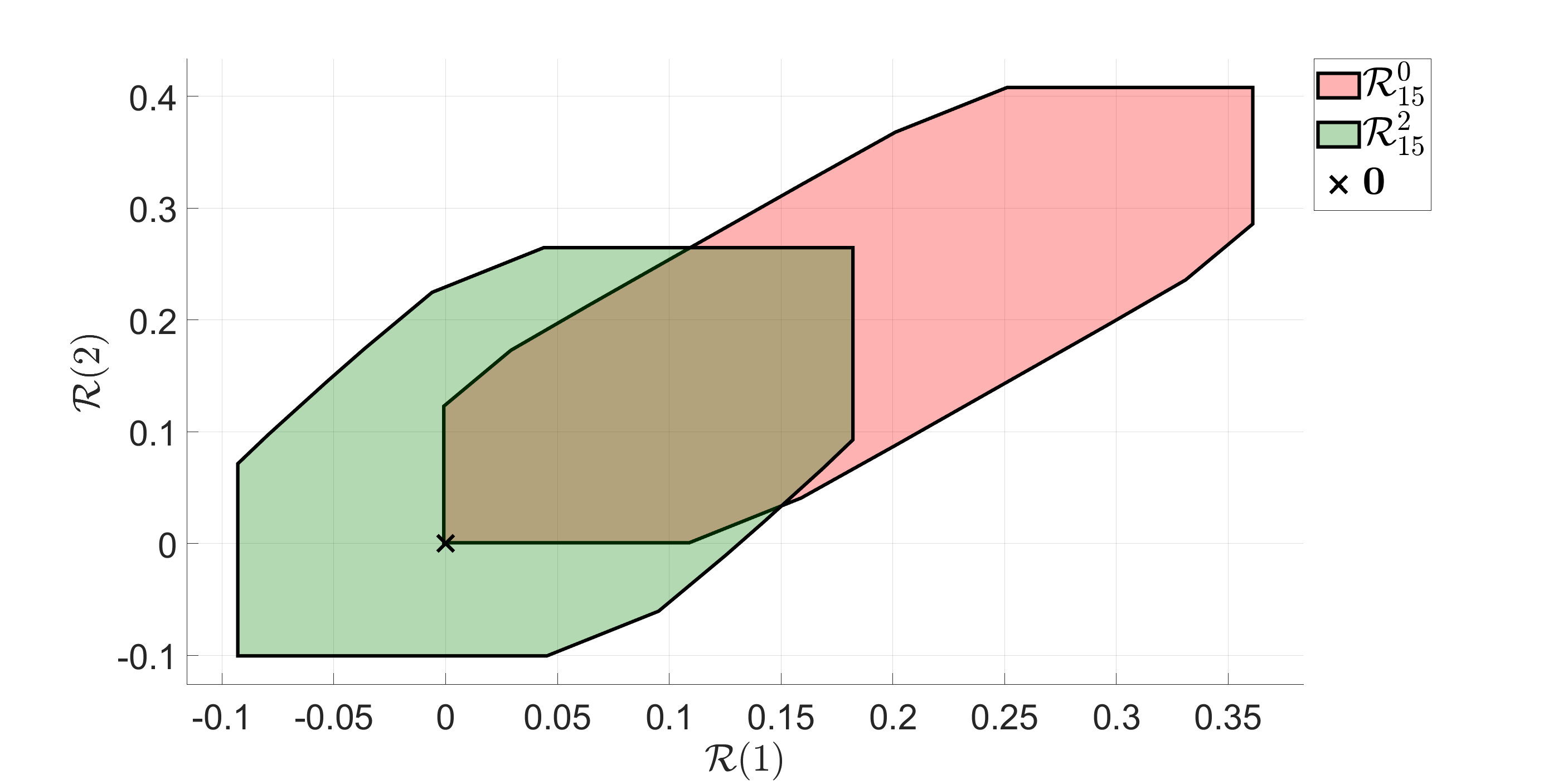}}
		\caption{PFD of the second fault $G_2$ with considering stability}
		\label{c_G2_PFD_Stability}
	\end{figure}
	
	In the third scenario, we compare the proposed method in this paper with the previous gain optimization method proposed in \cite{Xu2022CDC}. In the comparison, the same system parameters with those given in \cite{Xu2022CDC} are used, i.e., 
	\begin{align*}
		A=&\begin{bmatrix} 0.7 & 0.5 \\0 & 0.7 \end{bmatrix}, B=\begin{bmatrix} 1 & 0 \\0 & 1 \end{bmatrix},E=\begin{bmatrix} 0.5 & 0 \\0 & 0.5 \end{bmatrix},\\
		C=&\begin{bmatrix} 0.5 & 0.5 \\0 & 1 \end{bmatrix}, F=\begin{bmatrix} 0.1 & 0 \\0 & 0.1 \end{bmatrix}.
	\end{align*}
	In this example, we consider the healthy mode $G_0 = I$ and a fault interval $\mathbf{G}$ in the two actuators, i.e., 
	\begin{align*}
		\mathbf{G}=\begin{bmatrix} [0,0.9] & 0 \\0 & [0,0.9] \end{bmatrix}.
	\end{align*}
	
	The initial state, disturbances and noises are bounded by the following zonotopes: $\hat{\mathcal{X}}_0 =\langle \begin{bmatrix}\begin{smallmatrix} 0\\0\end{smallmatrix}\end{bmatrix} ,\begin{bmatrix}\begin{smallmatrix} 0.2 & 0 \\0 & 0.2 \end{smallmatrix}\end{bmatrix}\rangle$, $\mathcal{W}=\langle \begin{bmatrix}\begin{smallmatrix} 0\\0\end{smallmatrix} \end{bmatrix},\begin{bmatrix}\begin{smallmatrix} 0.1 & 0 \\0 & 0.1 \end{smallmatrix}\end{bmatrix}\rangle$ and $\mathcal{V}=\langle \begin{bmatrix}\begin{smallmatrix} 0\\0\end{smallmatrix} \end{bmatrix},\begin{bmatrix}\begin{smallmatrix} 0.1 & 0 \\0 & 0.1\end{smallmatrix} \end{bmatrix}\rangle$, respectively.
	
	Since \cite{Xu2022CDC} only considers fault detection, we also simplify the proposed method of this paper to only consider fault detection by only designing an observer and its gains for the healthy mode. It should be emphasized that the stability condition and method solving the indefinite quadratic objective function in \cite{Xu2022CDC} is different from those used in the proposed method of this paper. In order to make a fair comparison and without loss of generality, we change the stability condition and optimization problem solving method used in \cite{Xu2022CDC} to be the same with those used in the proposed method of this paper.
	
	In this simulation, a fault $G=\begin{bmatrix}\begin{smallmatrix}0.45 & 0\\ 0 & 0.45 \end{smallmatrix}\end{bmatrix} \in \mathbf G$ is injected into the system at time instant $k=29$. Since the fault detection performance is influenced by inputs, in order to comprehensively compare the current proposed method with the previous method in \cite{Xu2022CDC}, their fault detection performances are compared with different constant inputs. We choose different inputs $(u_1,u_2)\in \mathbb{P}_u\times \mathbb P_u$, where $\mathbb P_u = \left\{-0.26,-0.22,\cdots,0.22,0.26\right\}$ includes 14 elements with a gap of $0.04$ between two consecutive elements of $\mathbb P_u$. Therefore, $14\times 14$ different inputs in total are considered. Under each input, the current proposed method with the previous method in \cite{Xu2022CDC} are performed to detect the fault in at most $k_{max} = 71$ steps. 
	
	\figref{CDC_NEW_result_PFD} shows the detection time instants of different inputs using the current proposed method with the previous method in \cite{Xu2022CDC}, respectively. In \figref{CDC_NEW_result_PFD}, different colors of solid squares mean different ranges of detection time instants of different given input, where the red solid squares mean that the corresponding inputs need only 1 step to detect the given fault $G$, the orange solid squares mean that the corresponding inputs need 2 to 6 steps to detect the fault, the yellow solid squares mean that the corresponding inputs need 7 to 11 steps to detect the fault, the light blue solid squares mean that the corresponding inputs need 12 to 26 steps to detect the fault, the blue solid squares mean that the corresponding inputs need 27 to 41 steps to detect the fault, and the gray solid squares mean that the corresponding inputs need at least more 42 steps to detect the fault. Moreover, \figref{COMPARE_result_PFD} shows the comparison of the two methods, where the red solid squares mean the detection time of the current proposed method is earlier than those of the previous method in \cite{Xu2022CDC}, the blue solid squares mean opposite, and the yellow solid squares mean that the two methods have the same detection time. In general, \figref{COMPARE_result_PFD} shows that the current proposed method outperforms the previous method in \cite{Xu2022CDC} for most of given inputs in this example. Besides, although only one fault is used for the illustration in the third scenario, we have already tested more faults to compare the two methods. A general result is that the current proposed method can achieve effective performance improvements with respect to the previous method in \cite{Xu2022CDC}. However, due to the limit of space, only a fault is finally used for the illustration here. 
	
	\begin{figure}[htbp]
		\centering
		\subfigure[Ranges of FD time instants of the method in \cite{Xu2022CDC}]{\includegraphics[height=6cm,width=9cm]{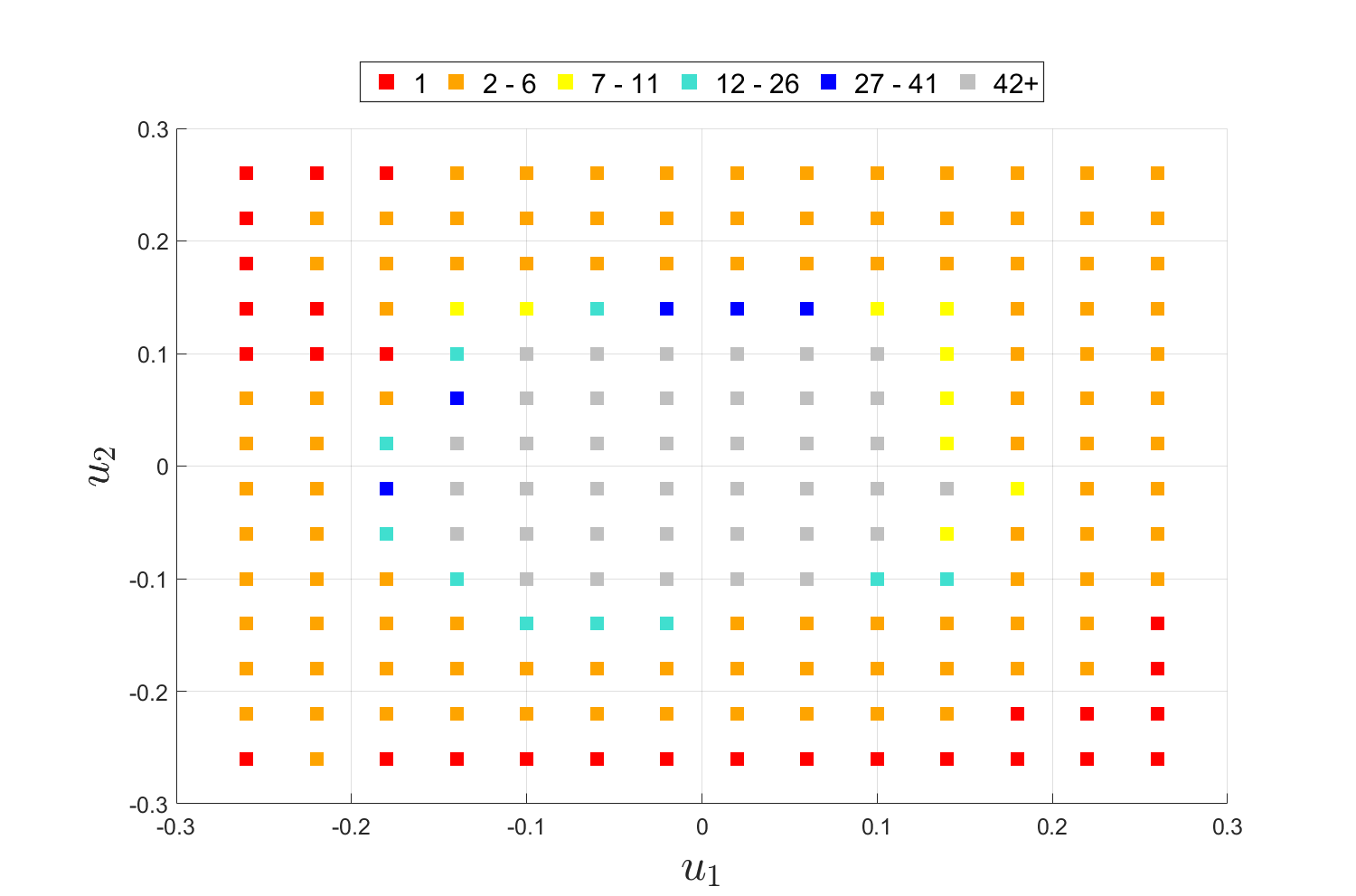}}
		\subfigure[Ranges of FD time instants of the proposed method]{\includegraphics[height=6cm,width=9cm]{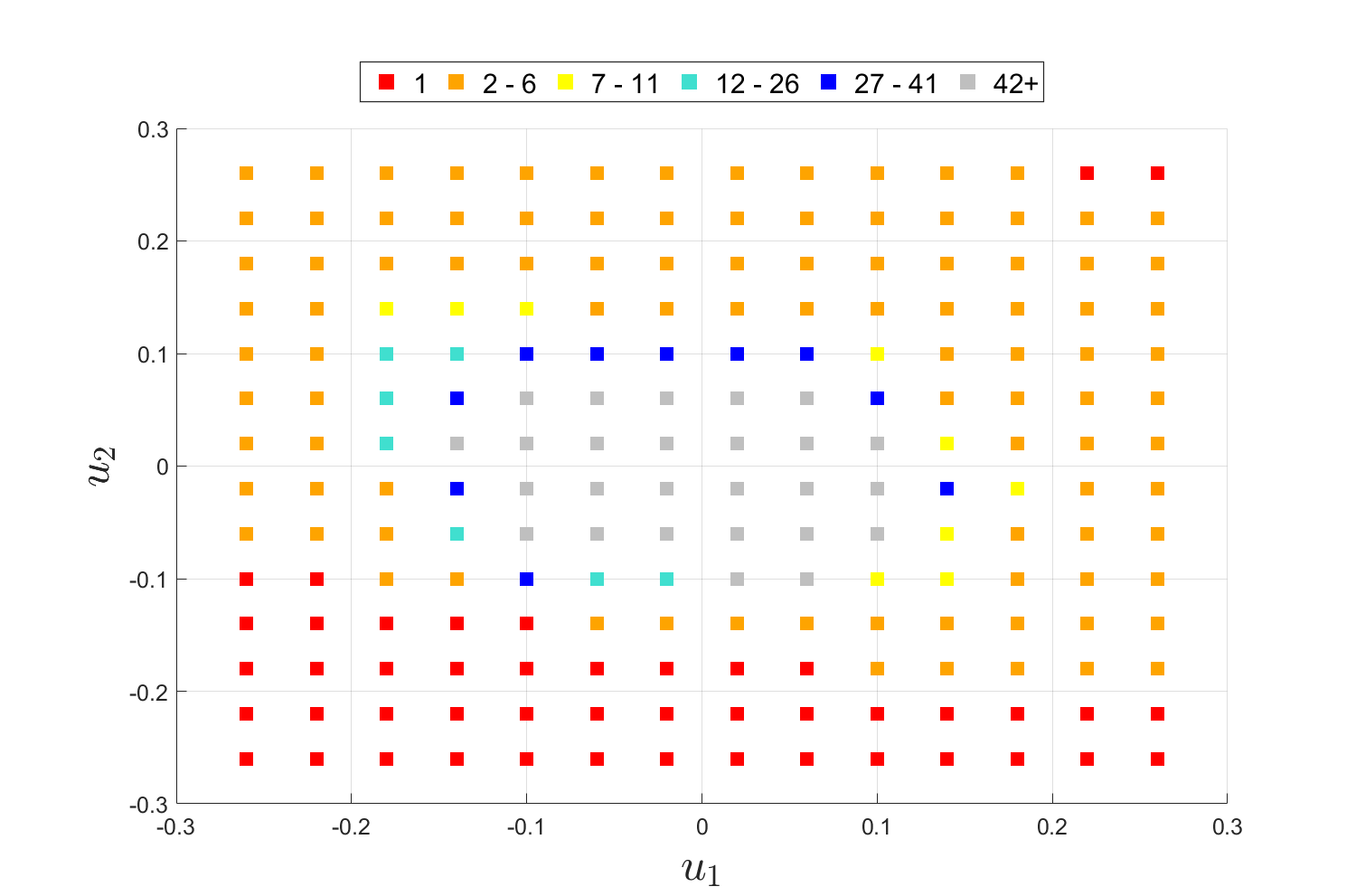}}
		\caption{FD time of different given inputs for both methods}
		\label{CDC_NEW_result_PFD}
	\end{figure}
	\begin{figure}[htbp]
		\centering
		\includegraphics[height=6cm,width=9cm]{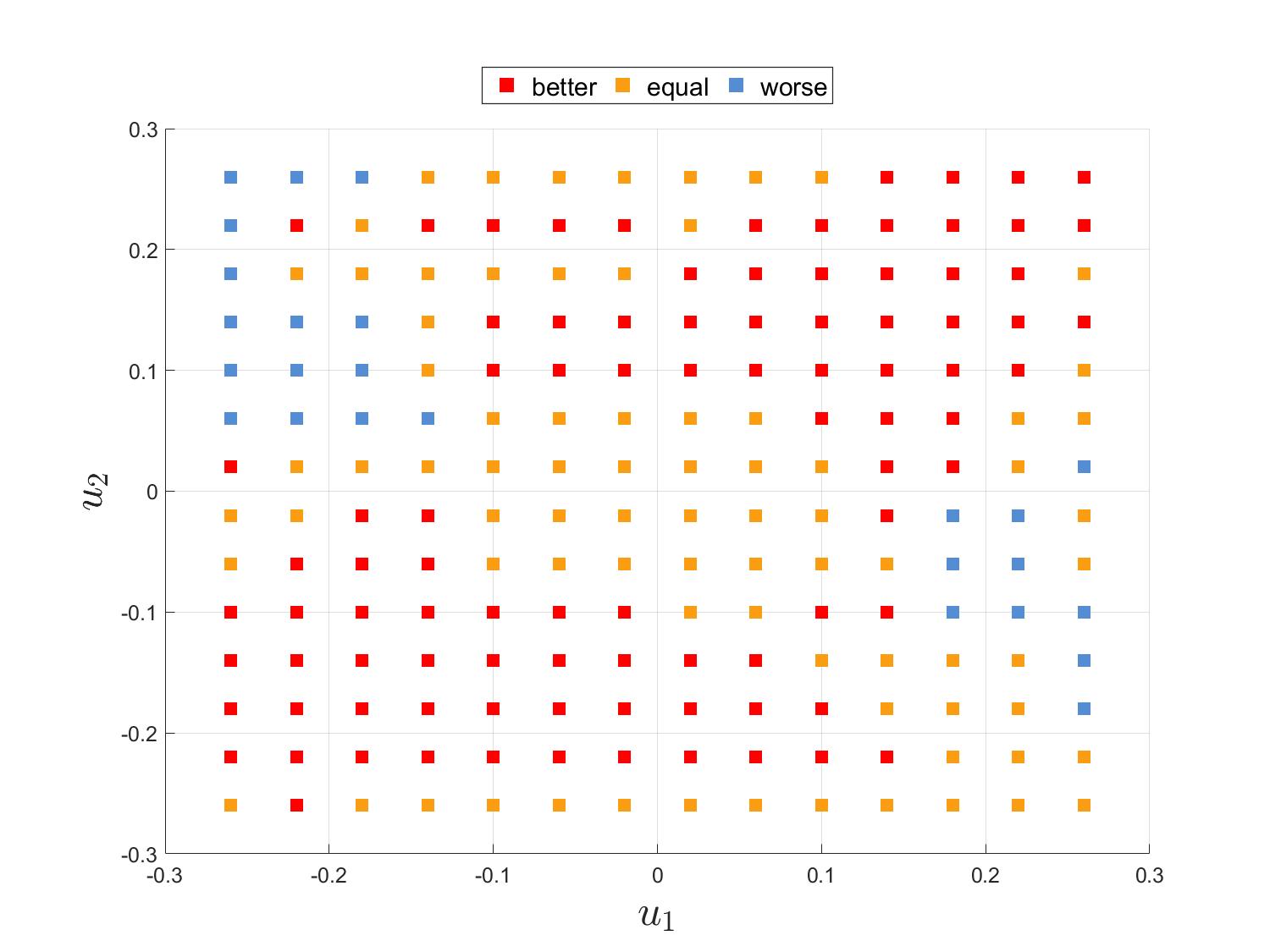}
		\caption{Comparison of FD time of the two methods}
		\label{COMPARE_result_PFD}
	\end{figure} 
	
	\subsection{Joint Gain and Input Design for AFD}
	\label{5_2}
	This subsection aims to illustrate the effectiveness of the proposed joint gain and input optimization method for AFD. We consider the same faults in the two actuators as shown in Subsection \ref{5_1} to show that the proposed joint gain and input optimization method can achieve better diagnosis performance than the single gain optimization method (i.e., the proposed PFD method in Section \ref{Section4_2}). In particular, a hard input constraint $\Vert u\Vert_2 \le 4$ is given, $\epsilon_2$ is set as $0.01$, and the initial state sets used in \eqref{SetModel} and \eqref{SetModel4} are given as $\mathcal{X}^0_0=\mathcal{X}^1_0=\mathcal{X}^2_0=\langle \begin{bmatrix}\begin{smallmatrix} 0.6\\0.6\end{smallmatrix} \end{bmatrix},\begin{bmatrix}\begin{smallmatrix} 0.1 & 0 \\0 & 0.1 \end{smallmatrix}\end{bmatrix}\rangle$. Consequently, the diagnosis results of the fault $G_1$ are shown in \figref{1_AFD_G1} where $\mathbf{0} \in \mathcal{R}^1_{5}$ and $\mathbf{0} \not\in \mathcal{R}^2_{5}$ are tested. This implies that the fault $G_1$ in the first actuator is diagnosed at $k=5$. Compared with the results in Subsection \ref{5_1}, it is shown that the proposed joint gain and input optimization method needs less time for diagnosis and has obvious advantages over the proposed PFD method in Section \ref{Section4_2}. Similarly, the diagnosis results of the fault $G_2$ are shown in \figref{1_AFD_G2} where $\mathbf{0} \not\in \mathcal{R}^1_{2}$ and $\mathbf{0} \in \mathcal{R}^2_{2}$ are tested as well. This also shows that the fault $G_2$ is diagnosed at $k=2$ and that the proposed joint gain and input optimization method has obvious advantages over the proposed PFD method in Section \ref{Section4_2}.
	\begin{figure}[htbp]
		\centering
		\subfigure[]{\includegraphics[height=4cm,width=8.5cm]{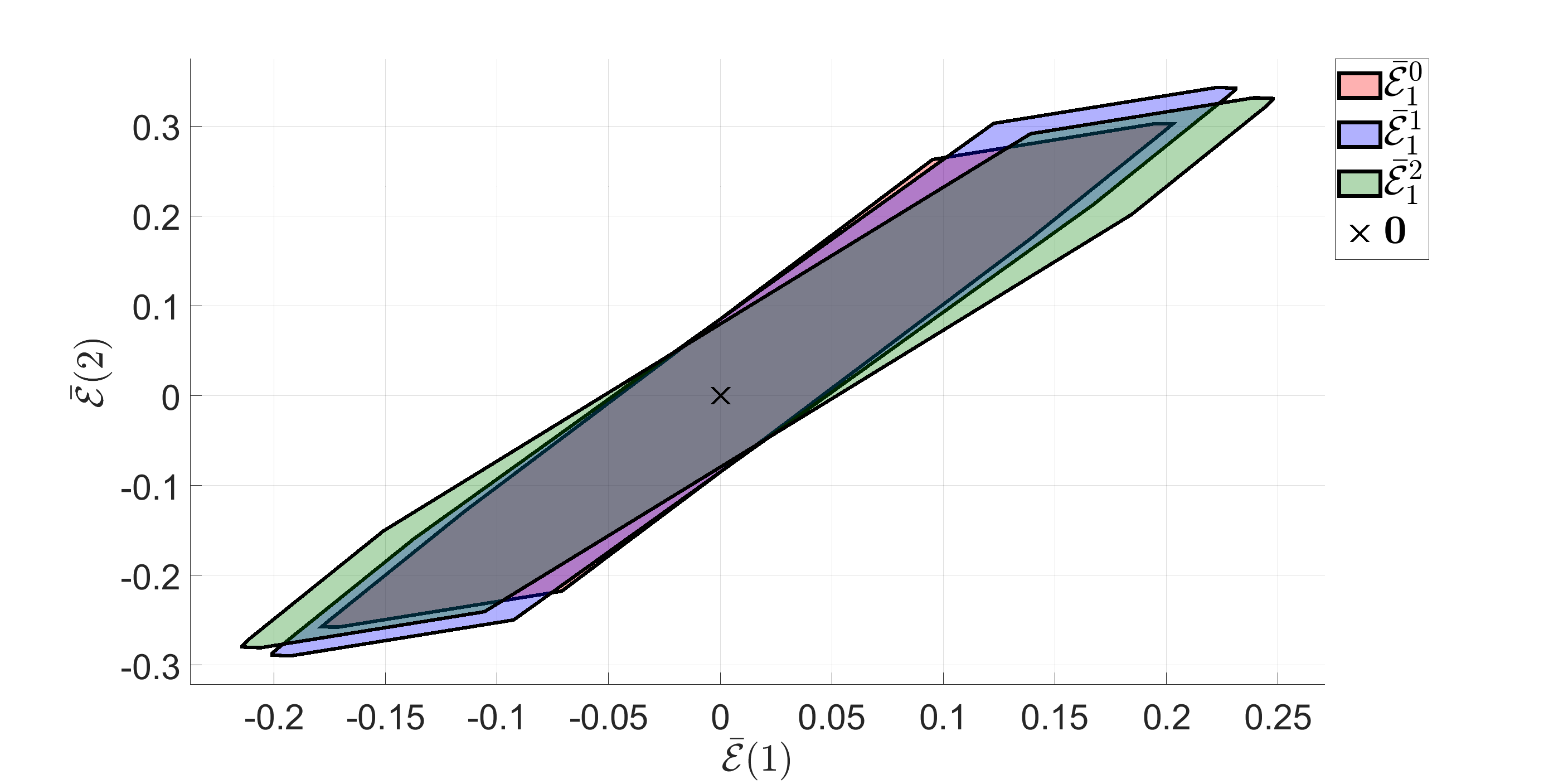}}
		\subfigure[]{\includegraphics[height=4cm,width=8.5cm]{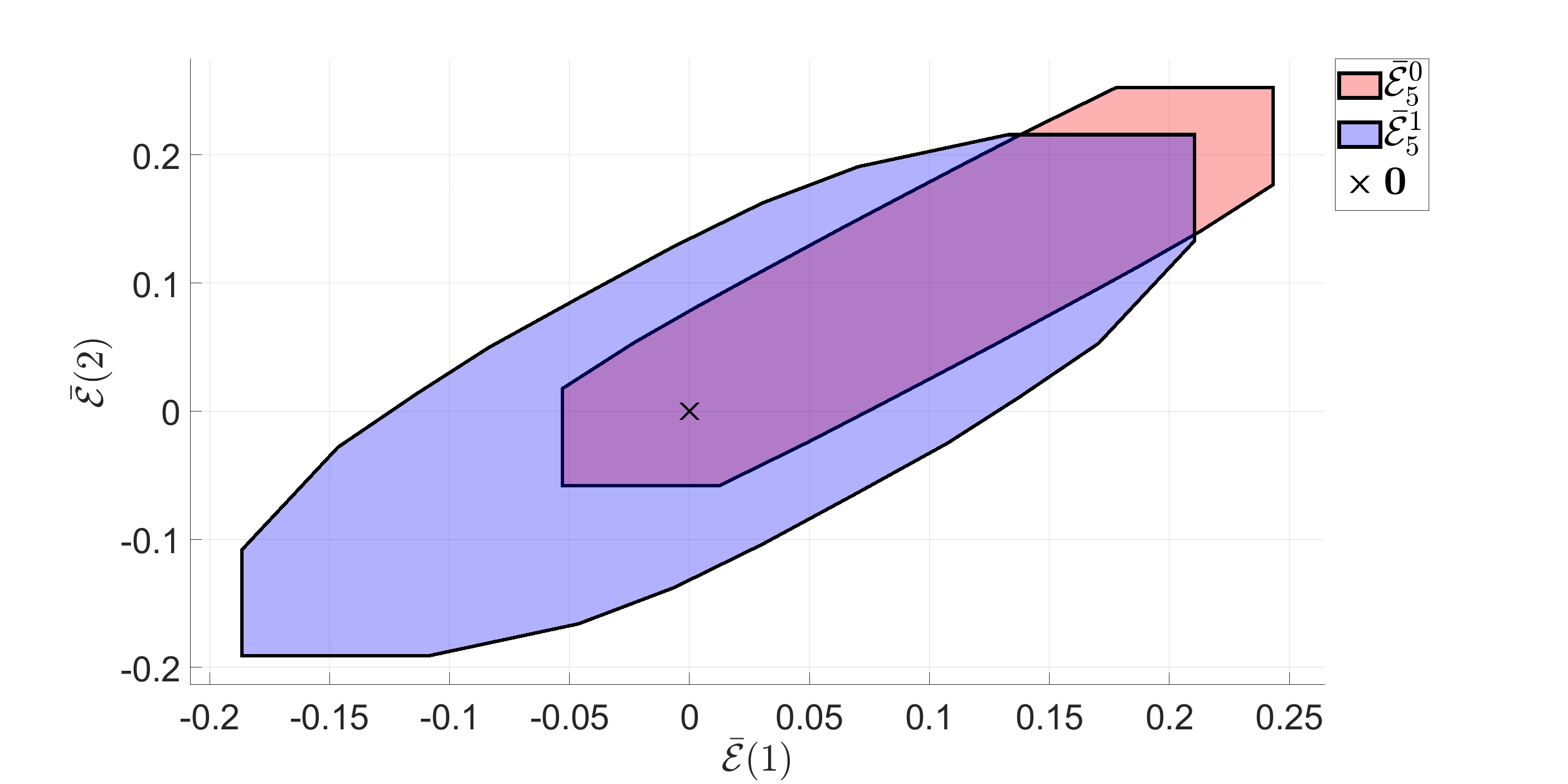}}
		\subfigure[]{\includegraphics[height=4cm,width=8.5cm]{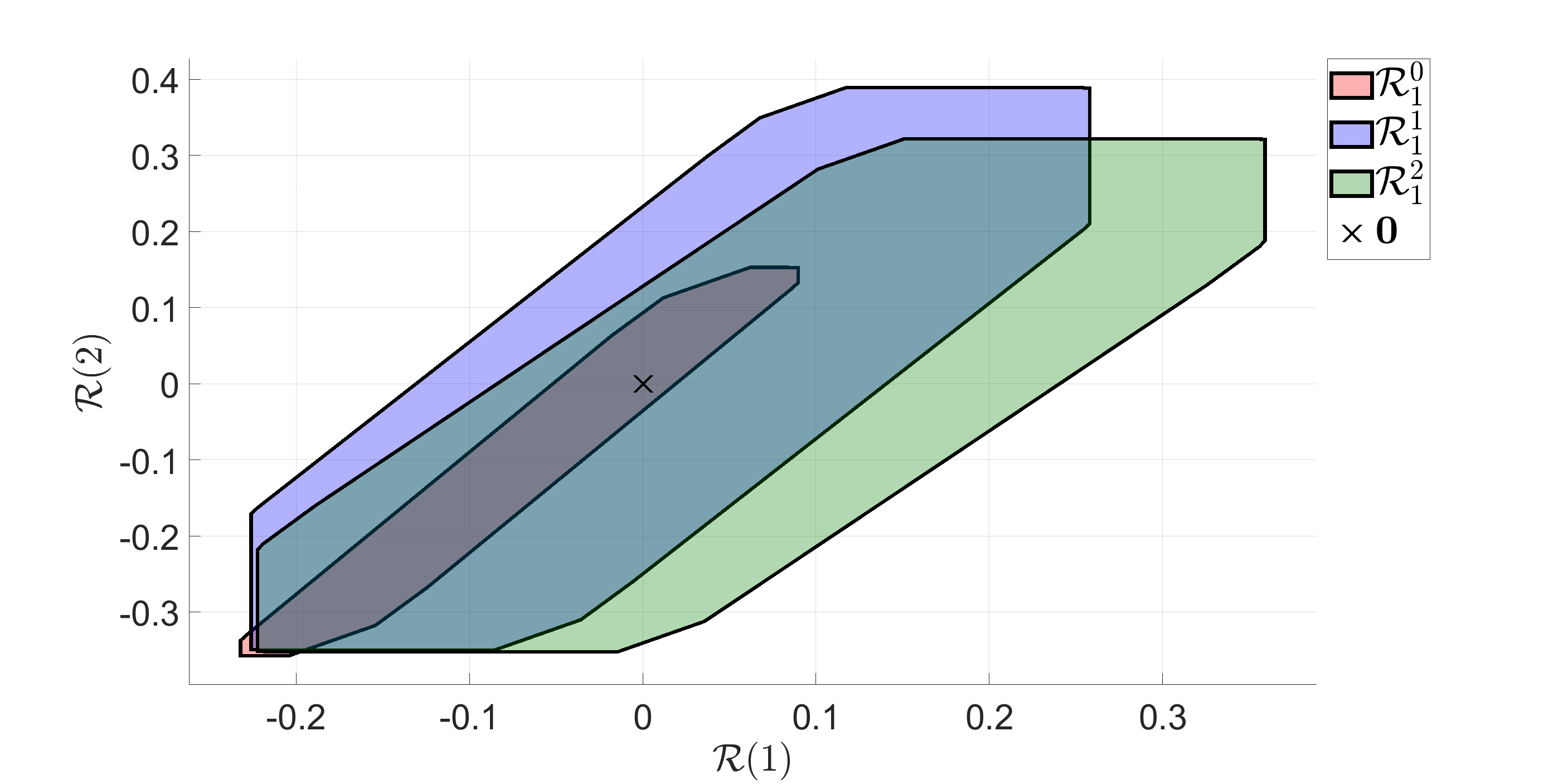}}
		\subfigure[]{\includegraphics[height=4cm,width=8.5cm]{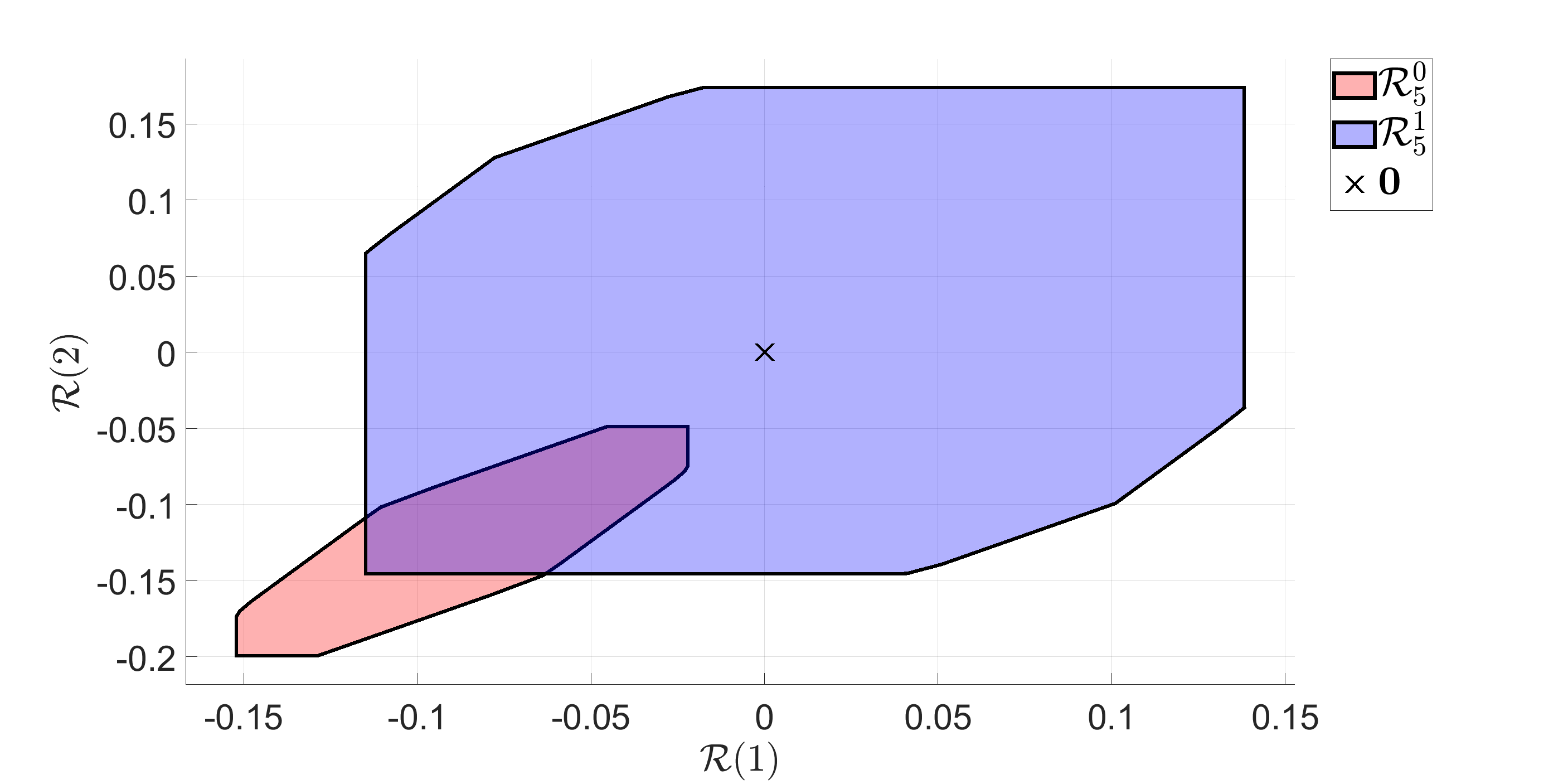}}
		\subfigure[]{\includegraphics[height=4cm,width=8.5cm]{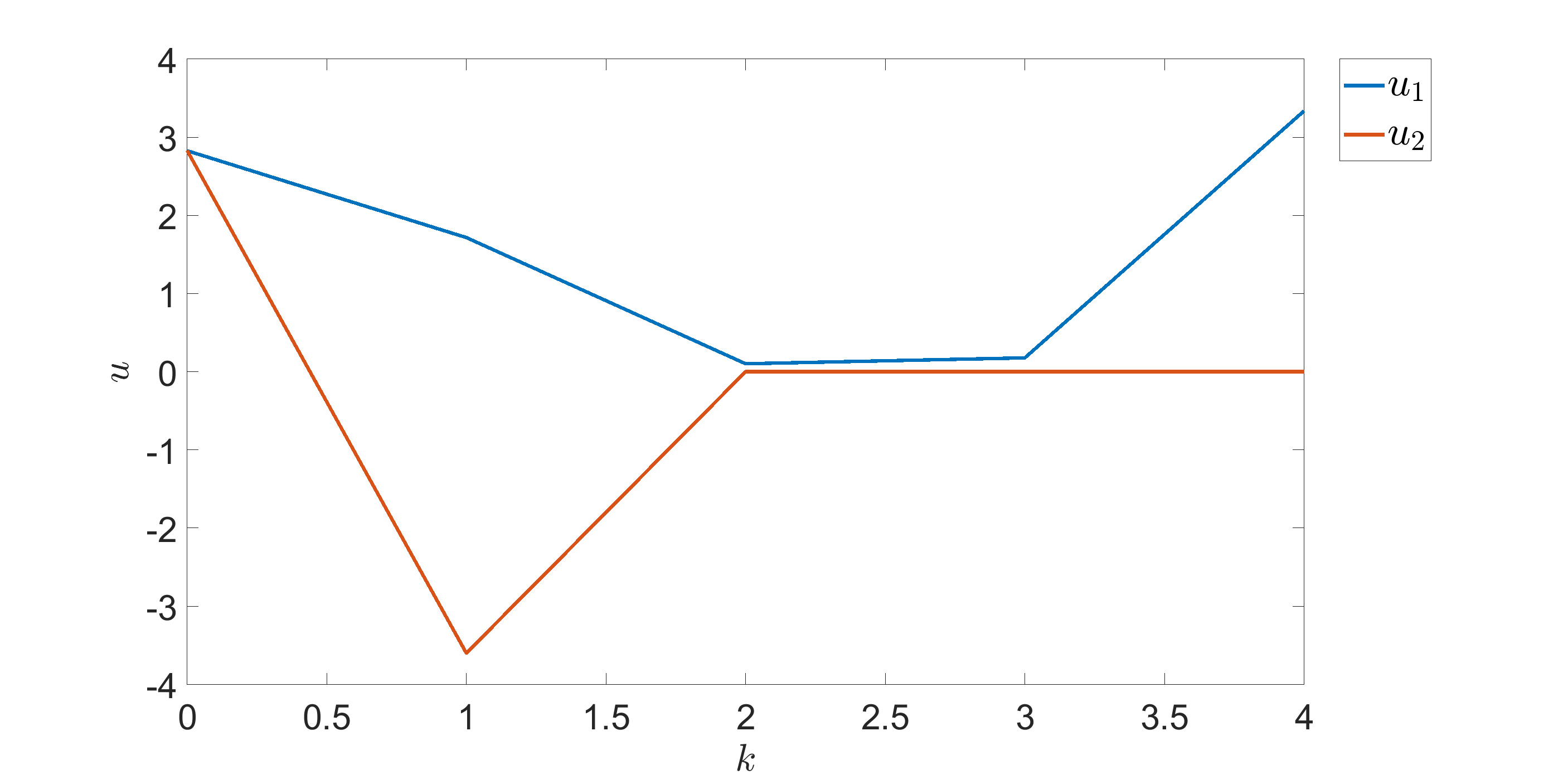}}
		\caption{AFD results of the first actuator fault $G_1$}
		\label{1_AFD_G1}
	\end{figure}
	\begin{figure}[htbp]
		\centering
		\subfigure[]{\includegraphics[height=4cm,width=8.5cm]{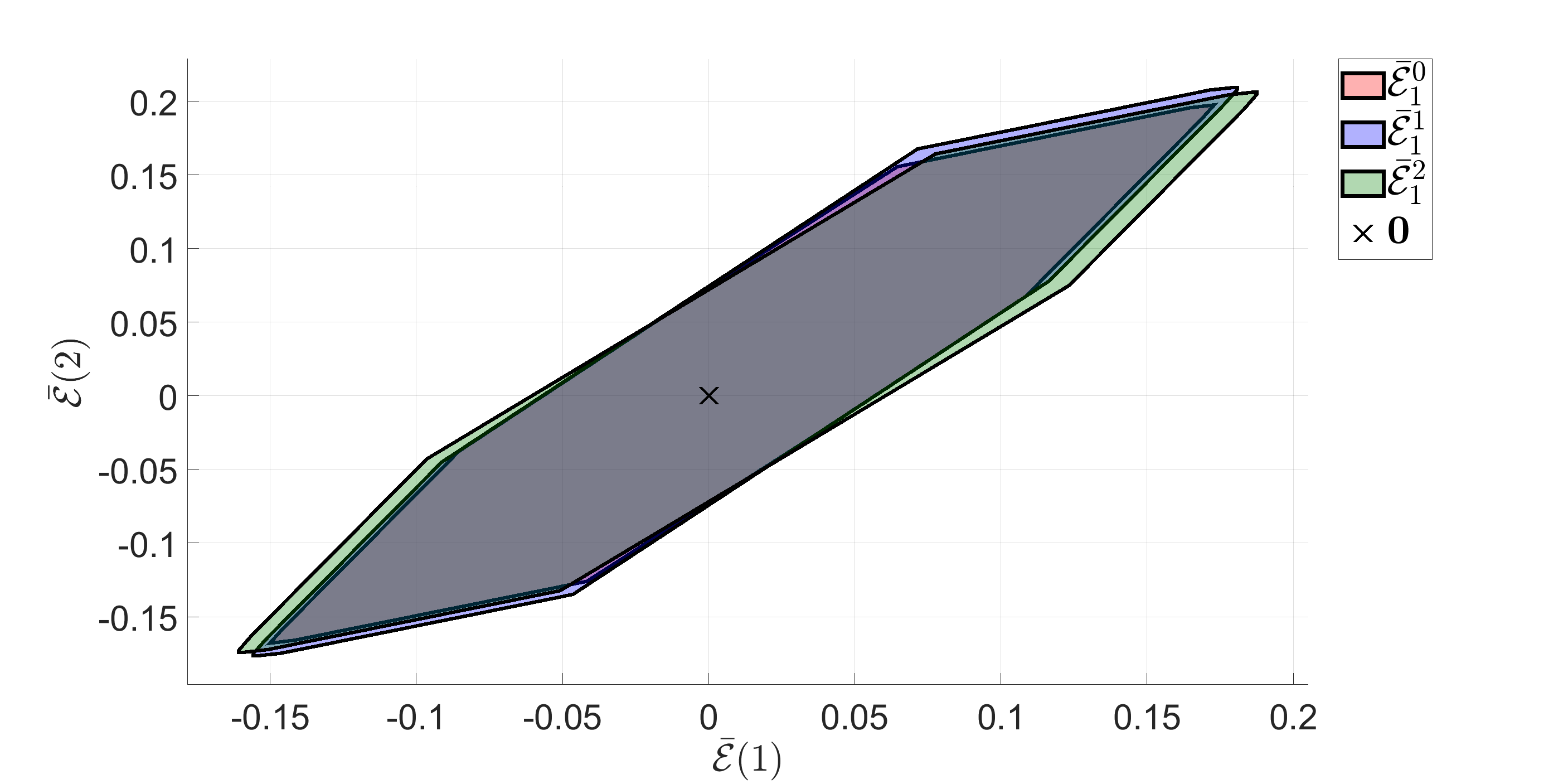}}
		\subfigure[]{\includegraphics[height=4cm,width=8.5cm]{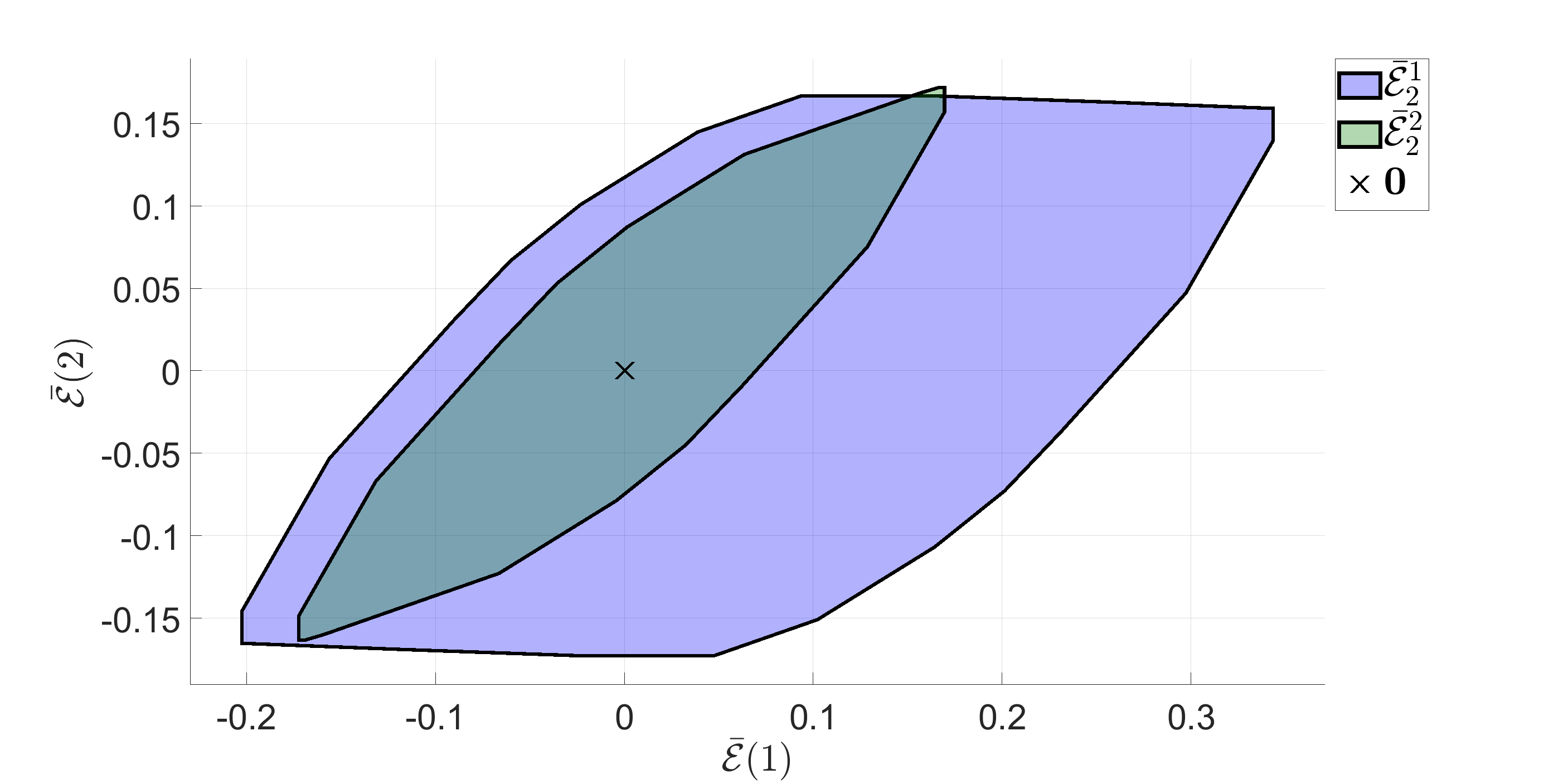}}
		\subfigure[]{\includegraphics[height=4cm,width=8.5cm]{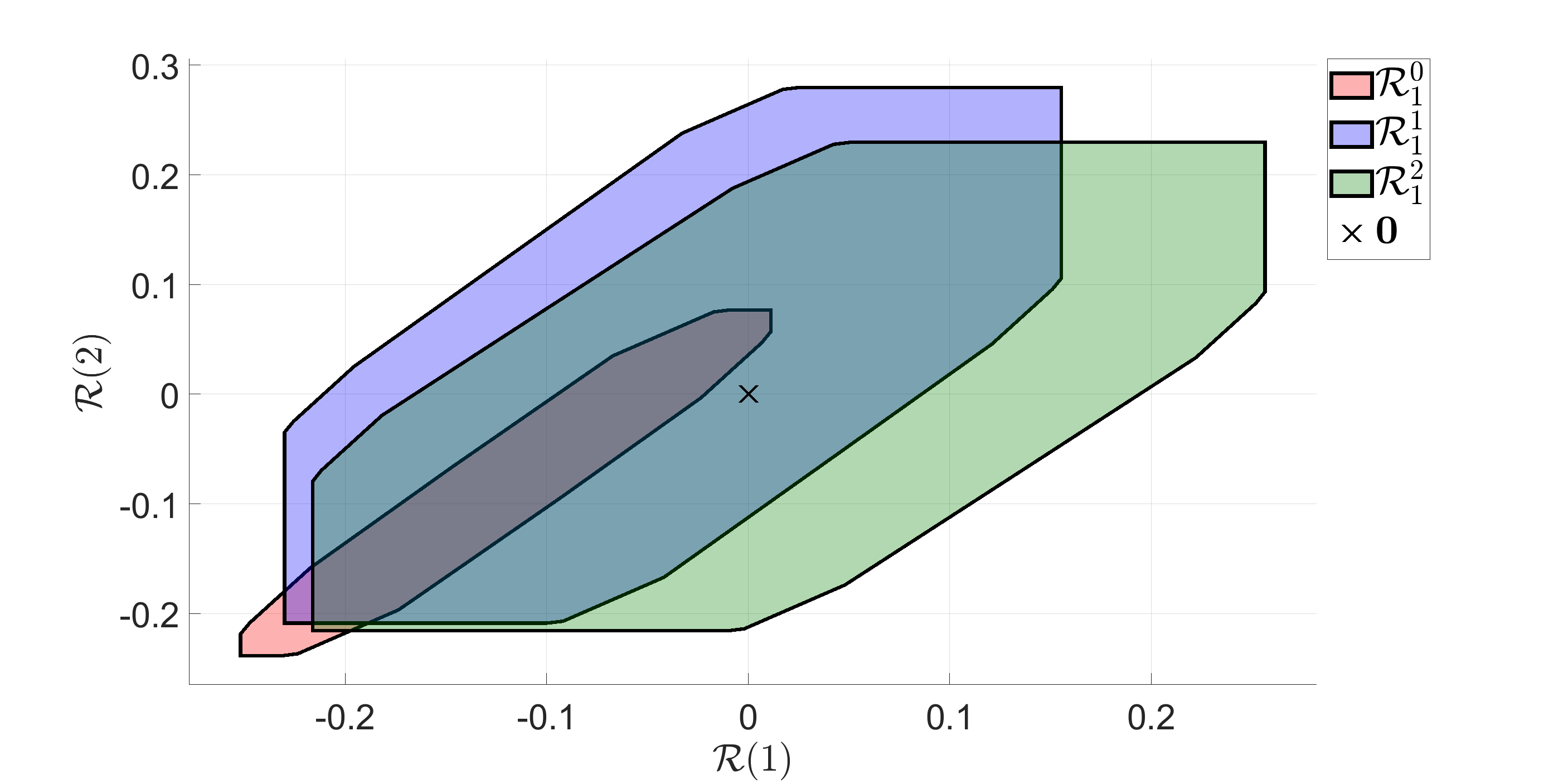}}
		\subfigure[]{\includegraphics[height=4cm,width=8.5cm]{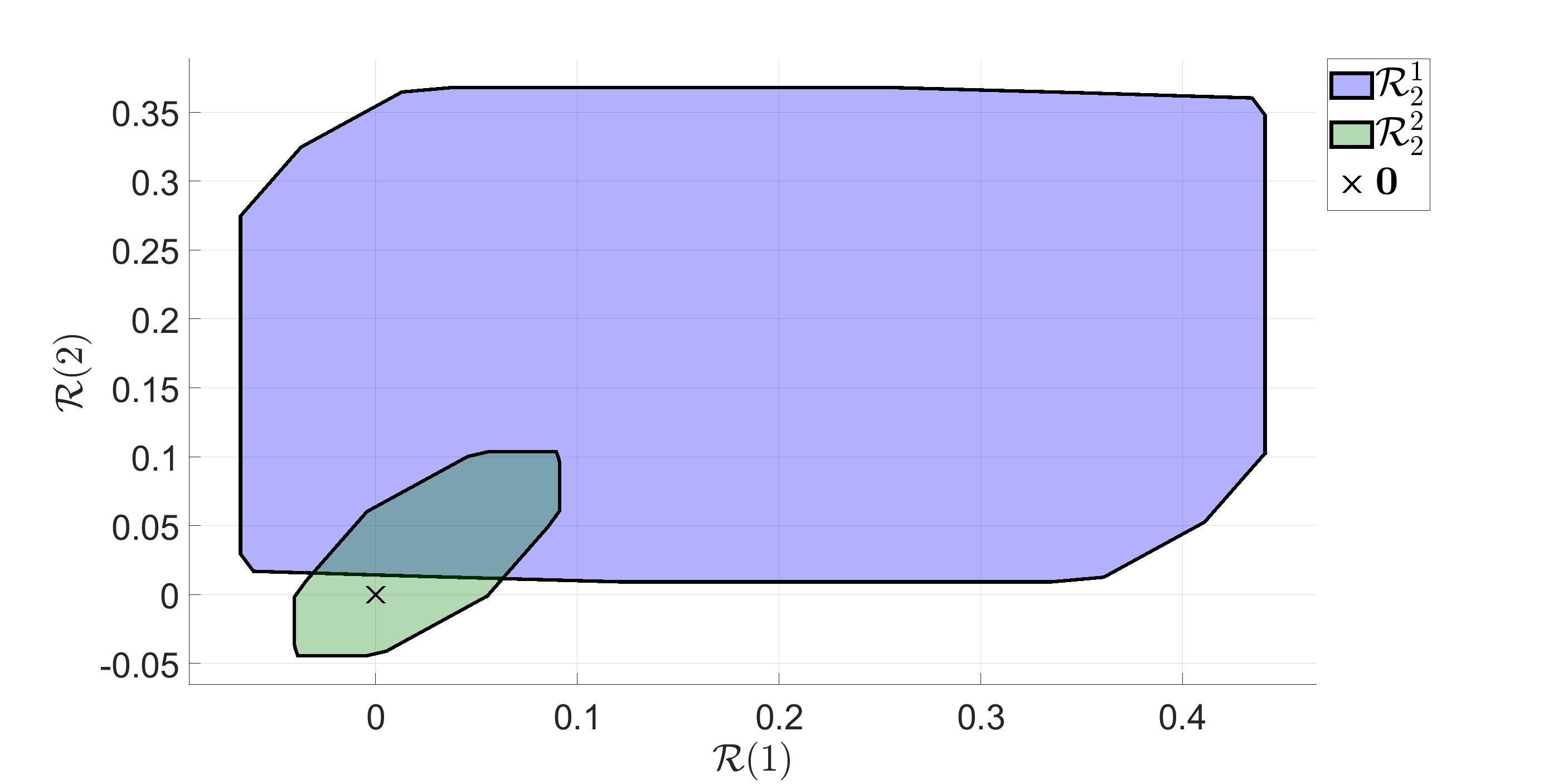}}
		\subfigure[]{\includegraphics[height=4cm,width=8.5cm]{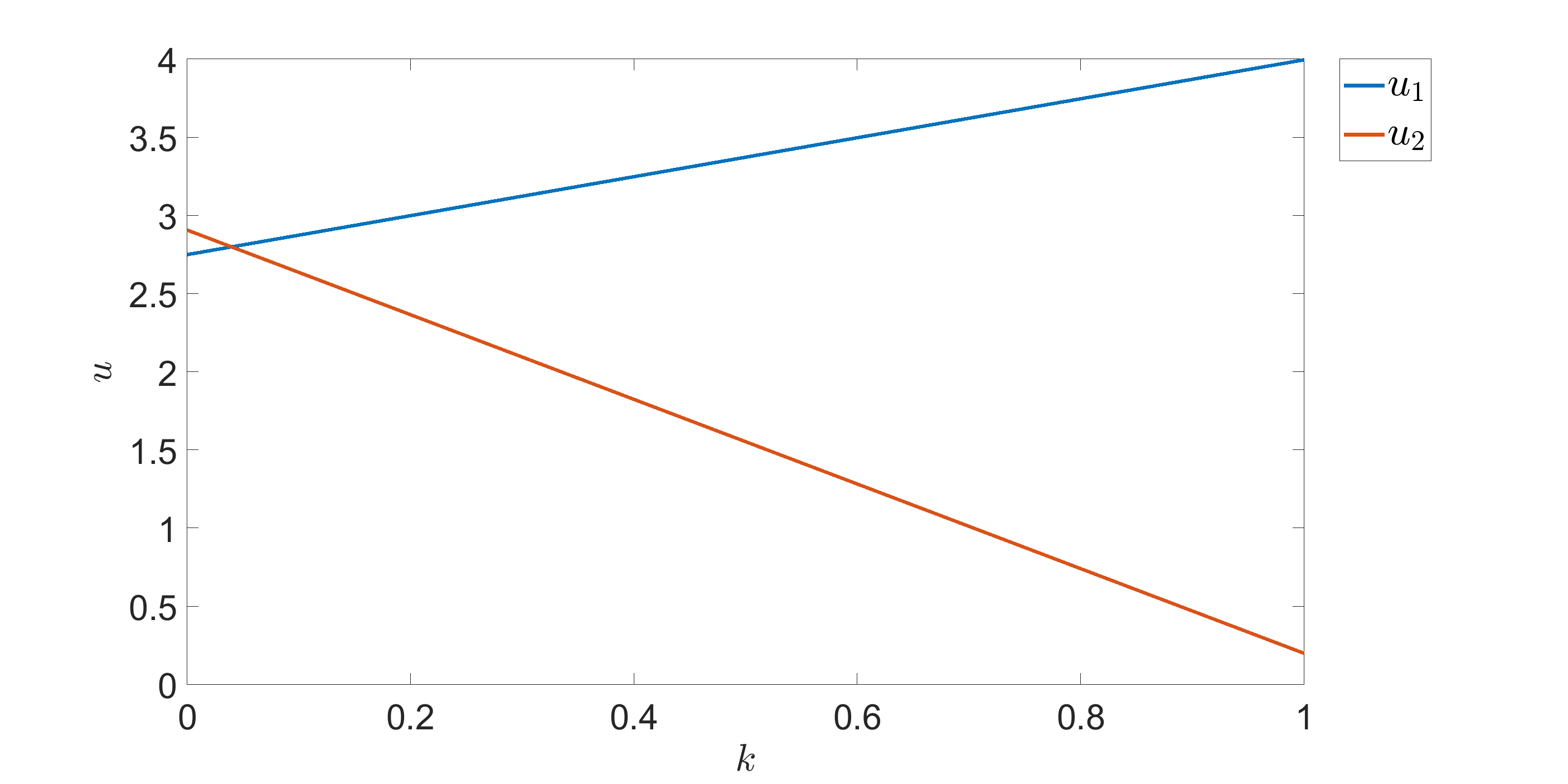}}
		\caption{AFD results of the first actuator fault $G_1$}
		\label{1_AFD_G2}
	\end{figure}
	
	Due to that the simulations above just compare the proposed joint gain and input optimization method with the proposed PFD method in Section \ref{Section4_2} using only one given input, in order to make a more convincing comparison, we do more comparisons between the proposed joint gain and input optimization method with the proposed PFD method using more sampling inputs. The system parameters keep the same with those used in Subsection \ref{5_1}. We test different inputs for the proposed PFD method and record the fault isolation time instant of each test. And then compare them with the results of the proposed joint gain and input optimization method for AFD. The comparisons of diagnosis results are shown in \figref{AFD_PFD_compare}. The red solid squares mean the proposed PFD method has shorter fault isolation time, the yellow solid squares mean the fault isolaton time of the two methods are the same, the green solid squares mean the proposed PFD method can isolate the fault within 21 time instants but the fault isolation time is longer than that of the proposed joint gain and input optimization method, and the blue solid squares mean the proposed PFD method cannot isolate the fault within 21 time instants at all. In \figref{AFD_PFD_compare}, it is shown that the proposed joint gain and input optimization method can achieve better FD performances than the proposed PFD method in most of cases, which also illustrates the effectiveness of the proposed joint gain and input optimization method. 
	
	\begin{figure}[htbp]
		\centering
		\subfigure[Diagnosis of the first fault $G_1$]{\includegraphics[height=6cm,width=9cm]{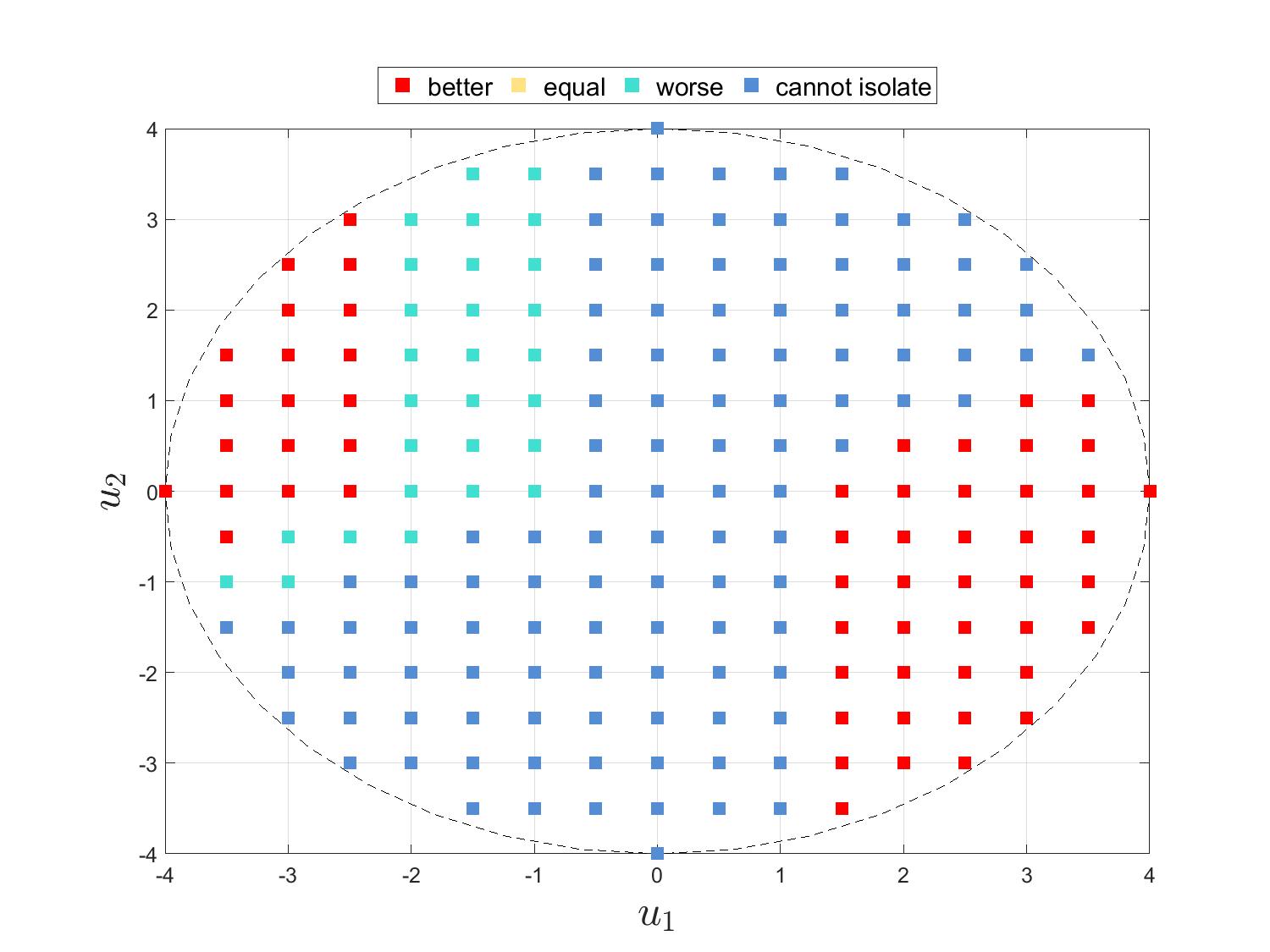}}
		\subfigure[Diagnosis of the second fault $G_2$]{\includegraphics[height=6cm,width=9cm]{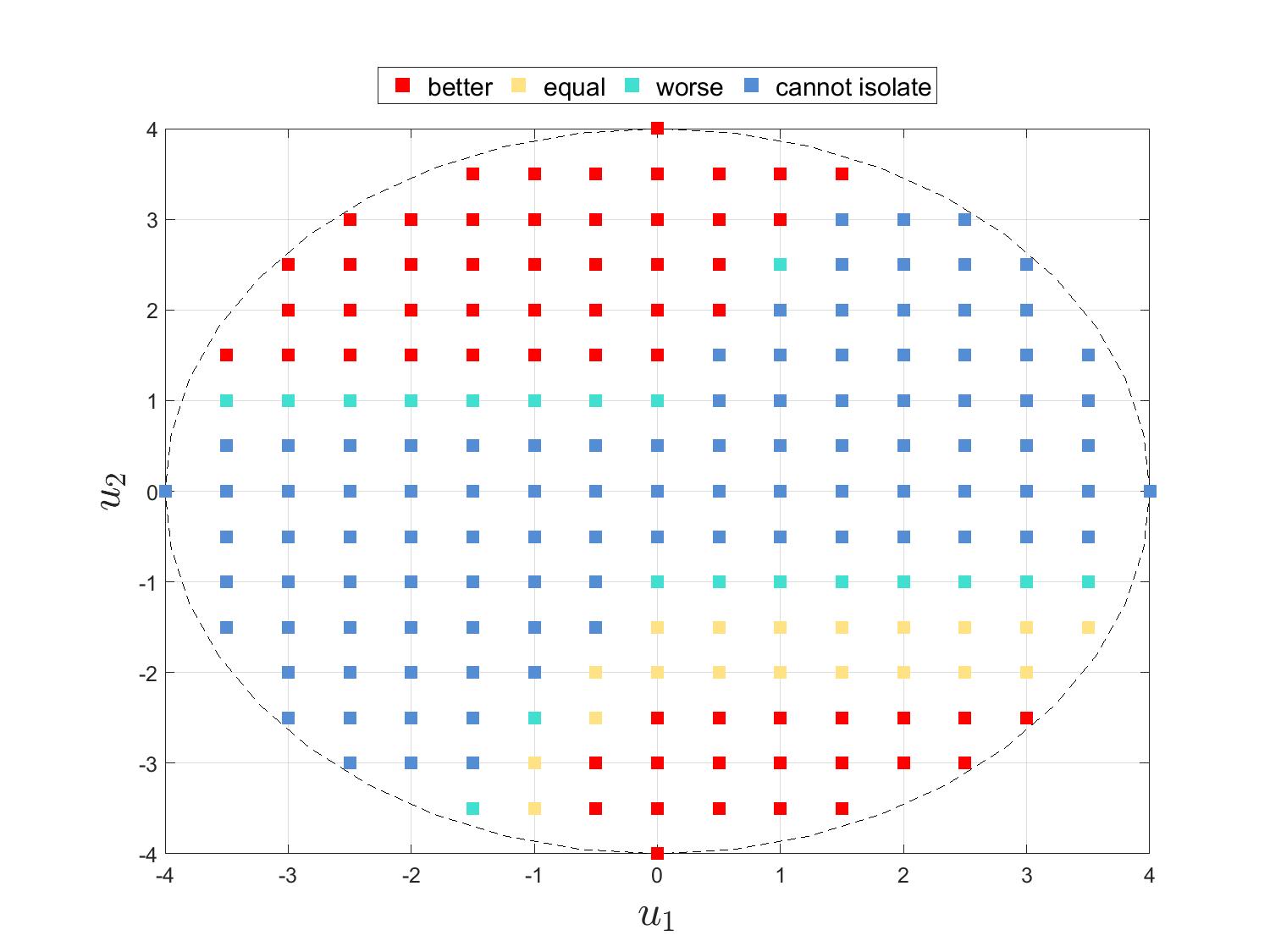}}
		\caption{Comparison of the methods in Sections \ref{Section4_2} and \ref{Section4_3}}
		\label{AFD_PFD_compare}
	\end{figure}
	
	\begin{remark}
		The comparisons shown above do not mean that the proposed joint gain and input optimization method absolutely outperforms the proposed PFD method in Section \ref{Section4_2}. Since the performance of the PFD method is affected by the given system inputs, if an appropriate input for the proposed PFD method happens to be injected into the system, it is also possible for the proposed PFD method to outperform the proposed joint gain and input optimization method in some cases. Besides, due to that the simulation process generates a large number of observer gains, this results in that we do not have enough space to display the gains. This is the reason why the designed gains are not shown here.   
	\end{remark}
	
	\section{Conclusions}
	In this paper, a new FD performance specification named the excluding degree of the origin from a zonotope is proposed. Based on the excluding degree, a new optimal FD gain design method is proposed for SVOs, which can be applied for both PFD and AFD methods using a bank of SVOs. Furthermore, using the excluding degree, the joint design problem of observer gains and inputs is solved for the observer-based AFD framework to sufficiently exploit its FD potential. In our future research, on one hand, more efficient optimization methods will be further explored to solve the formulated optimization problems to reduce computational complexity. On the other hand, less conservative stability conditions will be further investigated for SVOs.

	\section*{Acknowledgement}
		Yuxin Sun is acknowledged for some contributions in Section \ref{Section4_3} and the simulations. Bo Min is acknowledged for \propref{FD9}. This work is supported by the Shenzhen Science and Technology Program, China (No.JCYJ20210324132606015), the National Natural Science Foundation of China (Grant No. 62373161), and the Australian Research Council via the Discovery Early Career Researcher Award (DE220100609).
	
	\bibliographystyle{IEEEtran}
	\bibliography{refs}

\end{document}